\documentclass[11pt]{article}
\pdfoutput=1
\usepackage[margin=1in]{geometry}
\usepackage{amsmath,amsthm,amssymb,stmaryrd,enumitem,xspace,titlesec,comment}
\usepackage[usenames, dvipsnames]{xcolor}
\usepackage{mathpazo}

\DeclareMathAlphabet{\mathcal}{OMS}{cmsy}{m}{n}

\usepackage{algorithm}
\usepackage[noend]{algpseudocode}
\makeatletter
\renewcommand{\ALG@beginalgorithmic}{\footnotesize}
\makeatother
\usepackage{caption}
\usepackage{url}
\usepackage[colorlinks, linkcolor=BrickRed, citecolor=blue]{hyperref}
\usepackage{cleveref}
\crefname{appendix}{Section}{Sections}

\newlength{\squeezelen}
\makeatletter

\newenvironment{longversion}{
  \setcounter{result}{0}
  \setcounter{algorithm}{0}
  \newcommand\squeeze[1]{}
  \newcommand\longonly[1]{##1}
  \let\shortonly\@gobble
}{}
\makeatother

\def\clap#1{\hbox to 0pt{\hss#1\hss}}

\newcommand{\eps}{\varepsilon}

\newcommand{\FF}{\mathbb{F}}

\newcommand{\cA}{\mathcal{A}}
\newcommand{\cC}{\mathcal{C}}

\newcommand{\cF}{\mathcal{F}}
\newcommand{\cG}{\mathcal{G}}
\newcommand{\cI}{\mathcal{I}}

\newcommand{\cT}{\mathcal{T}}

\newcommand{\cV}{\mathcal{V}}
\newcommand{\cX}{\mathcal{X}}
\newcommand{\bx}{\mathbf{x}}
\newcommand{\tdelta}{\widetilde{\delta}}
\newcommand{\tO}{\widetilde{O}}

\renewcommand{\ge}{\geqslant}
\renewcommand{\le}{\leqslant}
\renewcommand{\geq}{\geqslant}
\renewcommand{\leq}{\leqslant}
\newcommand{\etal}{{et al.}\xspace}
\newcommand{\PP}{\mathsf{P}}
\newcommand{\NP}{\mathsf{NP}}
\newcommand{\Cov}{\mathit{Coverer}}
\newcommand{\Cv}{\mathit{Cov}}
\newcommand{\Sol}{\mathit{Sol}}
\newcommand{\Opt}{\mathit{Opt}}
\newcommand{\Alt}{\mathit{Alt}}
\newcommand{\Bac}{\mathit{Backup}}

\newcommand{\ang}[1]{\langle{#1}\rangle}
\newcommand{\ceil}[1]{\lceil{#1}\rceil}
\newcommand{\floor}[1]{\lfloor{#1}\rfloor}
\newcommand{\ed}[2]{\llbracket{#1}:{#2}\rrbracket}
\newcommand{\proc}[1]{\textsc{#1}}

\newcommand{\ontopt}[2]{\genfrac{}{}{0pt}{1}{#1}{#2}}
\newcommand{\class}[1]{\ang{#1}}

\newtheorem{theorem}{Theorem}[section]
\newtheorem{result}{Result}
\newtheorem{observation}[theorem]{Observation}
\newtheorem{lemma}[theorem]{Lemma}

\theoremstyle{definition}
\newtheorem{definition}[theorem]{Definition}

\DeclareMathOperator{\poly}{poly}

\DeclareMathOperator{\R}{R}

\newcommand{\setcover}{\textsc{set-cover}\xspace}
\newcommand{\partcover}{\textsc{partial-cover}\xspace}
\newcommand{\mpj}{\textsc{mpj}\xspace}
\newcommand{\idx}{\textsc{idx}\xspace}
\newcommand{\plr}{\textsc{plr}\xspace}

\title{Incidence Geometries and the Pass Complexity of Semi-Streaming Set Cover}
\author{Amit Chakrabarti%
\thanks{Department of Computer Science, Dartmouth College. 
Work supported in part by NSF Award CCF-1217375.}%
\and Anthony Wirth%
\thanks{Department of Computing and Information Systems, The University of Melbourne. 
Supported in part by an ARC Future Fellowship.}%
}
\date{}

\begin{document}

\maketitle

\begin{abstract} 
  \linespread{1.1}\selectfont    % sub/superscripts won't cause uneven line spacing
  Set cover, over a universe of size~$n$, may be modelled as a
  data-streaming problem, where the~$m$ sets that comprise the instance
  are to be read one by one.  A semi-streaming algorithm is allowed only
  $O(n \poly\{\log n, \log m\})$ space to process this stream. For each
  $p \ge 1$, we give a very simple deterministic algorithm that makes~$p$ passes
  over the input stream and returns an
  appropriately certified $(p+1)n^{1/(p+1)}$-approximation to the
  optimum set cover. More importantly, we proceed to show that this
  approximation factor is essentially tight, by showing that a factor
  better than $0.99\,n^{1/(p+1)}/(p+1)^2$ is unachievable for a $p$-pass
  semi-streaming algorithm, even allowing randomisation. In particular,
  this implies that achieving a $\Theta(\log n)$-approximation requires
  $\Omega(\log n/\log\log n)$ passes, which is tight up to the
  $\log\log n$ factor.

  These results extend to a relaxation of the set cover problem where we
  are allowed to leave an~$\eps$ fraction of the universe uncovered: the
  tight bounds on the best approximation factor achievable in~$p$ passes
  turn out to be $\Theta_p(\min\{n^{1/(p+1)}, \eps^{-1/p}\})$.

  Our lower bounds are based on a construction of a family of high-rank
  incidence geometries, which may be thought of as vast generalisations
  of affine planes. This construction, based on algebraic techniques,
  appears flexible enough to find other applications and is therefore
  interesting in its own right.
\end{abstract}

\thispagestyle{empty}
\addtocounter{page}{-1}
\newpage

\begin{longversion}
\section{Introduction}

The {\em set cover} problem is one of the most basic and well-studied
optimisation problems in computer science.  It features either directly
or in various guises in a wide array of applications, such as facility
location, information retrieval~\cite{AnagnostopoulosBBLMS15}, software
test selection, and tableau generation~\cite{GolabKKSY08}.  It is also
at the heart of a rich theory spanning approximation
algorithms~\cite{Vazirani-book} and computational complexity
theory~\cite{AroraBarak-book}, where efforts to understand the
complexity of set cover have led to interesting combinatorial and
mathematical interactions.  In this work, we consider set cover as a
``big data'' problem; specifically, we are concerned with space-efficient
algorithms for set cover in the well-established {\em data streaming}
model~\cite{Muthukrishnan05,BabcockBDMW02}. This setting has been
studied in several recent works, including Saha \& Getoor~\cite{SahaG09},
Emek \& Ros\'en~\cite{EmekR14}, and Demaine \etal~\cite{DemaineIMV14}.

An instance of set cover is given by a pair $(\cX,\cF)$, where~$\cX$ is
a finite universe with cardinality $|\cX| = n$ and $\cF \subseteq 2^\cX$
is a finite collection (multiset) of subsets of~$\cX$ with cardinality
$|\cF| = m$.
The pair $(\cX,\cF)$ satisfies the guarantee that the sets
in $\cF$ together cover $\cX$, i.e., $\bigcup_{S \in \cF} S = \cX$. A
candidate solution to the instance is a subcollection $\Sol \subseteq
\cF$; it is said to be {\em feasible} if $\bigcup_{S \in \Sol} S = \cX$.
Its {\em cost} is defined to be the cardinality $|\Sol|$.%
\footnote{In {\em weighted} set cover,
each set $S \in \cF$ has a cost or {\em weight} $w(S)
\ge 0$ and the cost of $\Sol$ is $\sum_{S \in \Sol} w(S)$.
The major contributions of this work being lower bounds, we focus on the
purely combinatorial setting, which is of course a strength for lower
bounds.}
%\footnote{A well-studied generalisation is {\em weighted} set cover,
%where each set $S \in \cF$ has an associated cost or {\em weight} $w(S)
%\ge 0$ and the cost of $\Sol$ is defined to be $\sum_{S \in \Sol} w(S)$.
%The major contributions of this work being lower bounds, we focus on the
%purely combinatorial setting, which is of course a strength for lower
%bounds.}
%
The desired goal is to find a feasible solution while keeping its cost small. A
feasible solution with minimum possible cost is said to be {\em optimal} and
its cost is called the {\em optimum cost} or {\em optimum value} of the
instance.  Henceforth we shall call this problem $\setcover_{n,m}$, or simply
$\setcover$.

It is well-known that finding an optimal solution to $\setcover$ is
$\NP$-hard~\cite{Karp72}; that finding an $\alpha$-approximate
solution---defined as a feasible solution whose cost is at most $\alpha$
times that of the optimum---is possible in polynomial time for $\alpha =
\ln n - \ln\ln n + \Theta(1)$~\cite{Slavik96}; and that doing so for
$\alpha < (1-\eps)\ln n$ is impossible unless $\NP =
\PP$~\cite{DinurS14}. Thus, for traditional Turing Machine computation,
the complexity of \setcover is essentially fully understood. However,
for genuinely huge instances of \setcover, additional considerations
become important: how will the data be accessed and how will it be
manipulated in a relatively small amount of working memory?

This motivates a careful study of the complexity of \setcover in a
data-streaming setting.  The instance $(\cX,\cF)$ is
presented as a stream consisting of the sets in~$\cF$, one at a time;
the universe~$\cX$ is known in advance, so we may assume that $\cX = [n]
:= \{1,2,\ldots,n\}$.  Representing an instance of $\setcover_{n,m}$
requires $\Theta(mn)$ bits in general.  Thus, in $\Theta(mn)$ bits of
space (working memory), we could simply run our favourite offline
algorithm.  The challenge is to work with {\em sublinear}---i.e.,
$o(mn)$---space. A {\em $p$-pass} algorithm may read its input stream up
to $p$ times; this parameter $p$, sometimes called the pass complexity,
ought to be a small constant, or perhaps $O(\log n)$.  Of course, in
addition to space and pass efficiency, we would also want our algorithms
to process each set quickly, with very simple operations and logic.

Since~$\Omega(n)$ space is required simply to certify
that a computed solution is feasible, we shall think of an
algorithm as highly space-efficient if it uses $\tO(n) := O(n
\poly\{\log n,\log m\})$ space.  Following a convention started with the
study of streaming graph algorithms~\cite{FeigenbaumKMSZ05}, and
continued in this context by Emek \& Ros\'{e}n~\cite{EmekR14}, we shall
call such an algorithm a {\em semi-streaming algorithm}. Emek \& Ros\'en
undertook a detailed study of {\em one-pass} semi-streaming algorithms
for \setcover, obtaining nearly tight bounds on the best approximation
ratio achievable by such algorithms. In this work, we provide tight
bounds for the multi-pass case, giving an almost complete understanding
of the pass/approximation tradeoff for semi-streaming algorithms. In
particular, this answers an open question explicitly raised by Saha \&
Getoor~\cite{SahaG09}.

\subsection{Our Results and Techniques}

A classic result of Johnson~\cite{Johnson74}, refined by
Slav\'ik~\cite{Slavik96}, gives a $(\ln n - \ln \ln n +
\Theta(1))$-approx\-imation to \setcover by a greedy algorithm. Given an
instance $(\cX,\cF)$, at each step, it adds to the current solution the
set from~$\cF$ that {\em contributes} most, i.e., covers the largest
number of as-yet-uncovered elements.
Notice that this can be
implemented as a semi-streaming algorithm, using one pass for each step,
but this leads to~$\Omega(n)$ passes, which is ridiculously expensive.
Saha \& Getoor~\cite{SahaG09} gave a different algorithm, which
guarantees an $O(\log n)$-approximation using only $O(\log n)$ passes.
Emek \& Ros\'{e}n~\cite{EmekR14} asked how good an approximation is
possible for a one-pass semi-streaming algorithm.  They showed that an
approximation ratio of $O(\sqrt n)$ is achievable and that the ratio
must be $\Omega(n^{1/2-\delta})$ for every constant $\delta > 0$;
\Cref{long_sec:related} adds some detail.  Our first result generalises their
upper bound, trading off additional passes for improved approximation.

\begin{result}[Formalised as \Cref{long_thm:ub}] \label{long_res:ub}
  In~$p$ passes, within semi-streaming space bounds, we can compute a
  $(p+1)n^{1/(p+1)}$-approximate solution to \setcover together with an
  appropriate ``certificate of coverage.''
\end{result}

\longonly{
The algorithm behind \Cref{long_res:ub} is a variant of the greedy approach
wherein each pass picks sets that contribute above some well-chosen
threshold for that pass, and the sequence of thresholds is geometrically
decreasing. This kind of thresholding is itself a variant of ideas
introduced by Cormode, Karloff, and Wirth~\cite{CormodeKW10} in a
non-streaming context. Our algorithm needs one final ``folding'' trick
that considers the final two thresholds in the sequence in a single
pass.
}

The Emek--Ros\'en algorithm solves a more general problem, with set
weights and a relaxed feasibility condition (partial coverage, which we
describe below).  For the basic combinatorial \setcover problem, our
algorithm nevertheless makes a (small) contribution even in the one-pass
case, with the simplicity of its logic as compared to Emek--Ros\'en: our
logic, being a variant of the basic greedy approach, is arguably easier
to implement and analyse.  But most importantly, this algorithm sets the
stage for our main result, which gets at the {\em pass complexity} of
the problem.

\begin{result}[Main result, formalised as \Cref{long_thm:lb}] \label{long_res:lb}
  In~$p$ passes, approximating the optimum of a \setcover instance to a
  factor smaller than $0.99\,n^{1/(p+1)}/(p+1)^2$ requires more than
  semi-streaming space. This applies even to the decision problem of
  distinguishing a small optimum value from a large one.
\end{result}

\Cref{long_res:ub,long_res:lb} together provide a near-complete understanding of
the power of each additional pass in improving the quality of an
approximate solution to \setcover. Saha \& Getoor had posed the problem
of obtaining this kind of tradeoff as an open question. \Cref{long_res:lb}
immediately implies that obtaining an $O(\log n)$-approximation under
semi-streaming space bounds requires $\Omega(\log n/\log\log n)$ passes,
almost matching the pass complexity of the Saha--Getoor algorithm (or,
for that matter, the algorithm behind our \Cref{long_res:ub}).

In establishing \Cref{long_res:lb}, we invent a family of novel combinatorial
structures that we call {\em edifices}. To explain these, we first
consider $p = 1$. In this case, an $\Omega(\sqrt n)$ bound follows from
a reduction from the \textsc{index} problem in communication complexity,
via set systems based on affine planes of finite order.%
\footnote{Emek \& Ros\'en also use affine planes, but differently, and
obtain an $\Omega(n^{1/2-\delta})$ bound versus our $\sqrt n/(4+\delta)$.}%
\shortonly{
A ``hard instance'' for one pass consists of one ``large'' set and many
``medium'' sets with very small pairwise intersections. The family of
lines in $\FF^2$, where $\FF$ is a finite field, gets us most of the way
towards these desired properties. To generalise this to $p > 1$, we
reduce from the multi-party communication problem
\textsc{pointer-jumping}. This calls for a more elaborate set system
with sets of many different sizes and a tree-like incidence structure,
plus a small-intersection property as before.
}%
\longonly{
A ``hard instance'' for one pass consists of a family of sets of two
different sizes: one ``large'' set and many ``medium'' sets with very
small pairwise intersections.  The family of lines in $\FF^2$, where
$\FF$ is a finite field, gets us most of the way towards the desired
properties. To generalise this to $p > 1$, we reduce from the
multi-party communication problem \textsc{pointer-jumping}. For this
reduction, we need a more elaborate set system with sets of many
different sizes (similar to contribution thresholds in the multi-pass
algorithms) and a tree-like incidence structure, plus a
small-intersection property as before.
}%
Very roughly, for $p = 2$, we start by considering quadric surfaces
inside $\FF^3$, and then lines in $\FF^2$ lifted onto these surfaces;
for higher $p$, we start with the appropriate extensions of these ideas
to higher-degree algebraic varieties. These varieties form a certain
incidence geometry---we call it an edifice---that is a vast
generalisation of affine planes. Bounding the sizes of certain pairwise
intersections between these varieties is the most technical part of this
work.

Following Emek \& Ros\'en, we also study {\em partial} set covers.  In
the $\partcover_{n,m,\eps}$ problem, an instance $\cI = (\cX,\cF,\eps)$
consists of $\cX$, $\cF$, and a parameter $\eps \in [0,1]$.  We require
a $(1-\eps)$-partial cover of~$\cX$: a collection $\Sol \subseteq \cF$
that covers at least $(1-\eps)|\cX|$ elements.  A solution $\Sol$ is
$\alpha$-approximate if $|\Sol| \le \alpha |\Opt|$, where $\Opt$ is a
minimum-cost \emph{total} set cover for $(\cX,\cF)$.

\begin{result}[Formalised as \Cref{long_thm:ub-partial,long_thm:lb-partial}] \label{long_res:partial}
  The smallest $\alpha$ for which a semi-streaming algorithm can compute
  an $\alpha$-approximate $(1-\eps)$-partial cover is in
  $\Theta_p(\min\{n^{1/(p+1)}, \eps^{-1/p}\})$. The lower bound applies
  to a decision problem of distinguishing a small total cover from a
  necessarily large partial cover.
\end{result}

The upper bound in \Cref{long_res:partial} builds on the one-pass
Emek--Ros\'en algorithm; thus we lose the extreme simplicity of the
algorithm behind \Cref{long_res:ub}, but gain the ability to handle {\em
weighted} instances.
The main contribution is again the
lower bound. It requires a reexamination of the edifices constructed for
establishing \Cref{long_res:lb} and proving that they satisfy additional
geometric properties. These properties then allow us to build new
edifices with different parameters that are suited to the problem at
hand. This construction shows the power of the axiomatic approach we
take in defining edifices.

We note in passing the minor result (formalised as \Cref{long_thm:comm}) that a tweak
to \Cref{long_res:lb} gives a rounds/ap\-prox\-imation tradeoff for a two-player
communication version of \setcover \`a la Nisan~\cite{Nisan02} and
Demaine \etal~\cite{DemaineIMV14}.

\shortonly{\input{related}}
\longonly{\subsection{Related Work} \label{long_sec:related}

The quantification of savings afforded by extra streaming passes dates
back to Munro \& Paterson~\cite{MunroP80}, who studied pass/space
tradeoffs for median-finding. This general topic remains
current~\cite{GuhaM07,ChakrabartiCM08,MagniezMN10,ChakrabartiCKM10}.

Efforts to understand the hardness of \setcover have led to many deep
insights and connections with various kinds of mathematics. Our
technical contributions continue this tradition. 
In the series of hardness-of-approximation results beginning with Lund
\& Yannakakis~\cite{LundY94,Feige98,Moshkovitz12}, recently culminated in
Dinur \& Steurer~\cite{DinurS14}, each result required new insights into
PCPs and parallel repetition; for details, see the latter paper and the
references therein.
Closer to this work, Nisan~\cite{Nisan02} initiated the study of
\setcover as a (two-player) communication problem and showed that, for
every constant $\delta > 0$, computing a $(\frac12-\delta)\log_2
n$-approximation to $\setcover_{n,m}$ requires $\Omega(m)$ randomised
communication.  His ``hard instances'' used $m \approx \exp(\sqrt n)$.
Nisan's original motivation was combinatorial auctions, but his result
can be interpreted in the data-streaming setting as saying that a
semi-streaming $(\frac12-\delta)\log_2 n$-approximation is impossible,
regardless of the number of passes.  Demaine \etal~\cite{DemaineIMV14}
showed that {\em deterministic} streaming algorithms achieving a
$\Theta(1)$-approximation require $\Omega(mn)$ space, thereby ruling out
sublinear-space solutions altogether. 

All of the above lower bounds have, at their core, some variant of an
old combinatorial construction: namely, that of a set system with the
so-called $r$-covering property~\cite{LundY94}.
Our own combinatorial constructions (of edifices) play an analogous role
in our lower bounds, but are quite different at a technical level. In
particular, they result in $\setcover_{n,m}$ instances where $m =
n^{\Theta(1)}$.  Their closest relative is the construction in Emek \&
Ros\'en~\cite{EmekR14} based on lines in an affine plane.

Turning to upper bounds, traditional (offline) approximation algorithms
for \setcover are discussed at length in Vazirani~\cite{Vazirani-book};
see also Slav\'ik~\cite{Slavik96} and the references therein. Alon
\etal~\cite{AlonAABN03} studied \setcover in an online setting,
focussing on competitive ratios rather than space considerations, but
under a fundamentally different input model: the sets are known in
advance and elements of the universe $\cX$ arrive in a stream. The
setting we study was first considered by Saha \& Getoor~\cite{SahaG09},
who called it ``set streaming.'' They gave a $4$-approximation algorithm
for \textsc{max-$k$-coverage}, the problem of choosing $k$ sets from the
stream so as to maximise the cardinality of their union. Iterating this
algorithm for $O(\log n)$ passes immediately gives an $O(\log
n)$-approximation for \setcover.
%%% We'll want to deal with this later.
%%% \mnote{T: I've probably noted this before, but there's \emph{almost} an implication here that Saha--Getoor is $O(n\log n)$ space for this task.}  
%%%
Cormode, Karloff, and
Wirth~\cite{CormodeKW10}, targeting external-memory efficiency,
developed a ``disk-friendly greedy'' (DFG) algorithm for \setcover. In
short, each step of DFG adds some set whose contribution is at least
$1/\beta$ times the maximum. As designed, DFG yields an $O(\log_\beta
n)$-pass, $(1 + \beta \ln n)$-approximate, $O(n \log n)$-space streaming
algorithm.

The single-pass semi-streaming setting was first, and thoroughly,
studied by Emek \& Ros{\'e}n~\cite{EmekR14}.  Indeed, their results
extend to \partcover, as well as item- and set-weighted variants. Their
algorithm, like ours, computes a {\em certificate of coverage} that
indicates, for each item, which set (if any) covers it: the implied
solution $\Sol$ covers a $1-\eps$ (weighted) fraction of $\cX$ and has
$w(\Sol) = O(\min\{1/\eps,\sqrt n\}w(\Opt))$.  
On the lower bound side, they prove that for every $\eps \ge 1/\sqrt n$,
a randomised semi-streaming algorithm that {\em certifies} an
(unweighted) $\alpha$-approximate $(1-\eps)$-cover must have $\alpha =
\Omega(1/\eps)$. Outputting only the sets in a solution (without a
certificate) still requires $\alpha = \Omega(\eps^{-1}\log\log n/\log
n)$. The still-weaker problem of approximating the optimum {\em value}
requires $\alpha = \Omega(n^{1/2-\delta})$ for every constant $\delta >
0$. Emek \& Ros\'en remark~\cite[footnote~3]{EmekR14} that they can show
this only for \setcover, and not for $(1-\eps)$-\partcover with $\eps
\gg 1/\sqrt n$. Compare these lower bounds with our
\Cref{long_res:lb,long_res:partial}, specialised to $p = 1$.

The main result of Demaine \etal~\cite{DemaineIMV14}, whose
deterministic lower bound we have discussed, is a {\em randomised}
sublinear-space, though not semi-streaming, algorithm for \setcover. It
achieves an $O(4^{1/\delta})$-pass, $O(4^{1/\delta}\rho)$-approximation
in $\tilde{O}(mn^\delta)$ space, where $\rho$ is the approximation ratio
of whatever offline \setcover algorithm we are prepared to run.

}

%!TEX root = main.tex

\section{A Simple Deterministic Multi-Pass Algorithm} \label{long_sec:ub}

\paragraph{Model of computation.}
An instance of $\setcover_{n,m}$
consists of sets $S_1,\ldots,S_m \subseteq [n]$, specified as a stream
of tokens $(i,S_i)$, where~$S_i$ is described in some reasonable way
(either as a list of its elements or as a characteristic vector) and~$i$
is the ID of~$S_i$. The IDs need not appear in the order $1,2,\ldots,m$.
The desired output is a set $\Sol \subseteq [m]$ consisting of the IDs
of sets that together cover~$[n]$, plus a certificate: an array
$\Cov[1\ldots n]$ in which, for each~$x$, $\Cov[x]$ is the ID of a set
that covers~$j$.  Strictly speaking, $\Sol$ is redundant because it can
be computed from $\Cov$, but keeping track of it explicitly aids
exposition.

Recall that a semi-streaming algorithm is allowed $\tO(n) := O(n
\poly\{\log n,\log m\})$ bits of space. This clearly suffices to
represent each of $\Sol$ and $\Cov$, which need only $\Theta(n\log m)$
bits, under the sensible assumption that $|\Sol| \le n$.  An ideal
semi-streaming algorithm for \setcover would use no more space than
this, asymptotically, and our \Cref{long_alg:basic,long_alg:folded} achieve
this space bound.

\longonly{\subsection{Algorithm and Analysis}}

\shortonly{\paragraph{Algorithm and analysis.}}
As promised, we begin by giving a very simple deterministic $p$-pass,
semi-streaming, ``progressive greedy'' algorithm that returns a
$(p+1)n^{1/(p+1)}$-approximation. The basic idea is that the first pass
is very conservatively greedy, picking a set into the solution iff its
{\em contribution} is at least some large number~$\tau_1$ (i.e., it
covers at least~$\tau_1$ as-yet-uncovered elements); the second pass
repeats this logic with a threshold $\tau_2 < \tau_1$, making it
slightly less conservative; and so on. Choosing suitable thresholds gets
us to a $p$-pass $pn^{1/p}$-approximation. This is the na\"ive version
of progressive greedy. Our final algorithm ``folds'' the last two passes
of this na\"ive version into a single pass, achieving the desired bound.

\algrenewcommand\algorithmicforall{\textbf{foreach}}
\begin{algorithm}[!ht]
  \caption{~~Na\"ive version of ``progressive greedy'' algorithm for \setcover, in $p$ passes \label{long_alg:basic}}
  \begin{algorithmic}[1]
    \Procedure{GreedyPass}{stream $\sigma$, threshold $\tau$, set $\Sol$, array $\Cov$}
    \ForAll{$(i,S)$ in $\sigma$}
	\State $C \gets \{x:\, \Cov[x] \ne 0\}$ 
	  \Comment{$C$ is the set of already covered elements}
        \If{$|S \setminus C| \ge \tau$} \label{long_line:contrib}
          \State $\Sol \gets \Sol \cup \{i\}$
          \ForAll{$x \in S \setminus C$}
	    $\Cov[x] \gets i$
	  \EndFor
        \EndIf
    \EndFor
    \EndProcedure
    \Statex \vspace{-5pt}

    \Procedure{ProgGreedyNaive}{stream $\sigma$, integer $n$, integer $p \ge 1$}
    \State $\Cov[1\ldots n] \gets 0^n$;~ $\Sol \gets \varnothing$
    \For{$j = 1 ~\textbf{to}~ p$}
      \Call{GreedyPass}{$\sigma, n^{1-j/p}, \Sol, \Cov$}
    \EndFor \vspace{-2pt}
    \State \textbf{output} $\Sol,\Cov$
    \EndProcedure
  \end{algorithmic}
\end{algorithm}

For ease of reading we have not optimised the per-token processing time
in our pseudocode. Clearly, in each pass, each set~$S$ can be processed
in $O(|S|)$ time in a RAM-style machine with word size $\Omega(\log m)$.

To analyse \Cref{long_alg:basic}, fix an arbitrary instance of
$\setcover_{n,m}$.  Each call to \proc{GreedyPass} makes a single pass
through $\sigma$, considering every set~$S$.  Note that the contribution
of~$S$ in such a pass is the quantity $|S \setminus C|$, computed in
\cref{long_line:contrib}, which is the number of \emph{new} elements that~$S$
covers.  Let $\Opt \subseteq [m]$ be an optimum solution.  For ease of
exposition we will pretend that $\Opt$ and $\Sol$ are collections of
sets from the input instance (they are in fact collections of IDs of
such sets).
%The next definition and lemma are key to analyzing the approximation
%guarantee of \Cref{long_alg:basic}.

\begin{definition} \label{long_def:bdd-pass}
  A $(\tau,\rho)$-bounded pass is a run of \proc{GreedyPass} with
  threshold~$\tau$ where, if~$C_0$ is
  the set of covered elements at the start of the pass, then for all~$S$
  in~$\sigma$ we have $|S \setminus C_0| \le \rho\tau$.
\end{definition}

\begin{lemma} \label{long_lem:bdd-pass}
  A $(\tau,\rho)$-bounded pass adds at most $\rho|\Opt|$ sets to
  $\Sol$.
\end{lemma}
\squeeze{\squeezelen}
\begin{proof}
  Put $D = [n]\setminus C_0$.  Each set in $\Opt$ includes at most
  $\rho\tau$ of the elements in~$D$, yet the sets in $\Opt$ together
  cover~$D$. Therefore $|\Opt| \ge |D|/(\rho\tau)$. Meanwhile, in this
  pass, each set added to $\Sol$ includes at least~$\tau$ elements
  of~$D$, so the pass adds at most $|D|/\tau \le \rho|\Opt|$ sets to
  $\Sol$.
\end{proof}

\begin{lemma} \label{long_lem:ub-basic}
  \Cref{long_alg:basic} is a $p$-pass semi-streaming $pn^{1/p}$-approximation
  algorithm for $\setcover_{n,m}$.
\end{lemma}
\squeeze{\squeezelen}
\begin{proof}
  The algorithm's correctness and $\tO(n)$ space bound are obvious, so
  we focus on the approximation ratio.  We claim that for each
  $j\in[p]$, the $j$th pass of \Cref{long_alg:basic} is $(n^{1-j/p},
  n^{1/p})$-bounded. 

  Let us prove this claim. Put $\tau_j = n^{1-j/p}$. For $j=1$, the
  precondition required by \Cref{long_def:bdd-pass} is trivially satisfied.
  For larger~$j$, consider an arbitrary set~$S$ in~$\sigma$ and
  let~$C_0$ be as in \Cref{long_def:bdd-pass}, for the $j$th pass. If~$S$ were
  added to $\Sol$ in an earlier pass, then $|S \setminus C_0| = 0$. If
  not, then by the logic of $\proc{GreedyPass}$, set~$S$'s contribution
  was less than $\tau_{j-1}$ during the $(j-1)$th pass. Since~$C_0$ is a
  superset of the set of elements that had been covered when~$S$ was
  processed in the $(j-1)$th pass, we have $|S \setminus C_0| <
  \tau_{j-1} = n^{1/p}\tau_j$.

  Having proved the claim, it follows from \Cref{long_lem:bdd-pass} that each
  pass adds at most $n^{1/p}|\Opt|$ sets to $\Sol$. Therefore, in the
  end we have $|\Sol| \le pn^{1/p}|\Opt|$, as required.
\end{proof}

\noindent
In fact, since the first pass adds at most $n^{1/p}$ sets, we have
$|\Sol| \le n^{1/p}(1+(p-1)|\Opt|)$.

% ---> paragraph, used to be subsection, which in turn used to be in tricolore.tex <---

\paragraph{Folding the last two passes.}

The final pass of \Cref{long_alg:basic} picks a set merely for making a {\em
nonzero} contribution. When there are at least two passes,
this final-pass logic can be ``folded into'' the penultimate pass as follows.
During the $p$th pass of a $p+1$-pass scheme, we run
\proc{GreedyPass} as usual and additionally, in parallel, run a second
instance of \proc{GreedyPass} with threshold~$1$ that builds an
alternate solution $\Alt$ (certified by a new array $\Bac$,
analogous to $\Cov$), starting from $\varnothing$.
Thus, $\Alt$ is the solution that a $1$-pass version of \Cref{long_alg:basic}
would have built. At the end of the penultimate ($p$th) pass, $\Sol$ might have
left some elements of~$\cX$ uncovered. We fix this by post-processing:
for each such element~$x$, we add to $\Sol$ the set in $\Alt$ that
covered~$x$; this information can be read from $\Bac$.%
\shortonly{

It is simple to show that the set $\Sol$ so computed is identical to
what \Cref{long_alg:basic} would have output after~$p+1$ passes. We have
therefore computed a $(p+1)n^{1/(p+1)}$-approximation using only~$p$ passes.
Equivalently, we have the following theorem (for pseudocode, see
\Cref{long_alg:folded}).
}%
\longonly{
\Cref{long_alg:folded} implements this very idea.
%but in~$p$ passes rather
%than $p+1$ (thus, it emulates a $(p+1)$-pass version of
%\Cref{long_alg:basic}).

%%% Interestingly, the last two passes of Algorithm~\ref{long_alg:basic} can be folded
%%% into one.
%%% In doing so, we mimic the algorithm of Emek and Ros\'{e}n in that they have
%%% some elements each covered by the cheapest set.
%%% To simplify the presentation, we compare out folded algorithm with
%%% a two-pass algorithm:
%%% we consider the 
%%% \setcover instance that remains after all but the last two passes as
%%% a new instance on the uncovered items and `non-empty' sets.
%%% 
%%% In Algorithm~\ref{long_alg:folded},
%%% An item is `black', in~$B$, when it is covered by a set that
%%% contributes at least~$\tau$, and is gray, in~$G$, when it is covered by a `smaller' set.

\algrenewcommand\algorithmicforall{\textbf{foreach}}
\begin{algorithm}[!ht]
  \caption{~~Progressive greedy algorithm for \setcover in~$p$ passes \label{long_alg:folded}}
  \begin{algorithmic}[1]
    \Procedure{ProgGreedy}{stream $\sigma$, integer $n$, integer $p \ge 1$}
    \State $\Cov[1\ldots n], \Bac[1\ldots n] \gets 0^n$;~ $\Sol, \Alt \gets \varnothing$
    \For{$j = 1 ~\textbf{to}~ p-1$}
      \Call{GreedyPass}{$\sigma, n^{1-j/(p+1)}, \Sol, \Cov$}
    \EndFor \vspace{-2pt}
    \State in parallel, 
      do \Call{GreedyPass}{$\sigma, n^{1-p/(p+1)}, \Sol, \Cov$}
      and \Call{GreedyPass}{$\sigma, 1, \Alt, \Bac$}
    \Statex
    \For{$x = 1 ~\textbf{to}~ n$} \label{long_line:postproc-begin}
      \Comment{Post-processing: elements not covered by $Sol$ will get covered by sets from $\Alt$}
      \If{$\Cov[x] = 0$}
        \State $\Sol \gets \Sol \cup \Bac[x]$
        \State $\Cov[x] \gets \Bac[x]$
      \EndIf
    \EndFor \label{long_line:postproc-end}
    \State \textbf{output} $\Sol,\Cov$
    \EndProcedure
  \end{algorithmic}
\end{algorithm}

\begin{lemma} \label{long_lem:folded}
  For every stream $\sigma$,
  the output of $\proc{ProgGreedy}(\sigma,p)$ in
  \Cref{long_alg:folded} is identical to that of
  $\proc{ProgGreedyNaive}(\sigma,p+1)$ in \Cref{long_alg:basic}.
\end{lemma}
\squeeze{\squeezelen}
\begin{proof}
  Fix an input stream~$\sigma$. Let~$\cA_1$ and~$\cA_2$ denote, respectively,
  the invocation of \Cref{long_alg:basic} as
  $\proc{ProgGreedyNaive}(\sigma,p+1)$ and the invocation of
  \Cref{long_alg:folded} as $\proc{ProgGreedy}(\sigma,p)$.  Let $\Sol_1$ be
  the value of $\Sol$ after~$p$ passes of~$\cA_1$.  It is immediate that
  $\Sol_1$ is also the value of $\Sol$ in~$\cA_2$ just before the
  post-processing loop in
  \crefrange{long_line:postproc-begin}{long_line:postproc-end}.

  Let $\Cv_1$ and $\Cv_2$ denote, respectively, the final {\em output}
  values of the array $\Cov$ in~$\cA_1$ and~$\cA_2$.  Let $C =
  \bigcup_{S \in \Sol_1} S$. Our above observation says that $\Cv_1[x] =
  \Cv_2[x]$ for all $x \in C$. It remains to prove that the same
  equality also holds for all $x \in [n] \setminus C$.  But this, too,
  is immediate from the observation that for each $x \in [n] \setminus
  C$, each of $\Cv_1[x]$ and $\Bac[x]$, and thus $\Cv_2[x]$ as well, is
  set to the earliest set in~$\sigma$ that contains~$x$.
\end{proof}
}

\begin{theorem} \label{long_thm:ub}
  There is a $p$-pass, $O(n\log m)$-space algorithm that, for every instance
  of $\setcover_{n,m}$, outputs a feasible solution $\Sol$ with $|\Sol|
  \le n^{1/(p+1)}(1 + p|\Opt|) \le (p+1)n^{1/(p+1)} |\Opt|$.
  \shortonly{\qed}
\end{theorem}
\longonly{
\squeeze{\squeezelen}
\begin{proof}
  This follows immediately by combining \Cref{long_lem:ub-basic} with
  \Cref{long_lem:folded}.
\end{proof}
}

\paragraph{Folding three passes?}
It is natural to wonder whether the above ``folding'' idea can be taken
further, achieving an even better pass/approximation tradeoff. As it
turns out, we cannot fold (the last) three passes into one. The most
convincing proof is the lower bound that we shall establish in
\Cref{long_sec:lb}. 

\longonly{
%But let us see, at an intuitive level, why this idea fails.
As designed, the algorithm cannot be sure what the
contribution of a set will be in a particular pass until it actually sees this
set in that pass.
In the last pass, however, we need only know that the contribution is non-zero:
after the penultimate pass, if $\Cov[x]=0$ and $\Bac[x]=i$, we know ``in
advance'' that set~$S_i$ has non-zero contribution.%
}

\longonly{\subsection{Tightness of Analysis} \label{long_sec:hardstream}}

\shortonly{\paragraph{Tightness of analysis.}}
The lower bound in \Cref{long_sec:lb} shows that the approximation ratio
guaranteed by \Cref{long_thm:ub} is asymptotically optimal for $p =
\Theta(1)$ passes. But if~$p$ is allowed to grow with~$n$, then there remains
a small $\Theta(p^3)$ discrepancy between that upper bound and the lower bound
we shall eventually prove in \Cref{long_thm:lb}. 
\longonly{
We can, however, prove that our
analysis of the approximation guarantee of \Cref{long_alg:basic} is tight.
}
\shortonly{
We can, however, show---as is done by the following theorem, proved in
\Cref{long_sec:hardstream}---that our analysis of the approximation
guarantee of \Cref{long_alg:basic} is tight.
}

\begin{theorem} \label{long_thm:anal-tight}
  For each integer $p \ge 2$ and~$q$ large enough, there is an instance
  $\cI_{n,m}$ of $\setcover_{n,m}$, with $n = q^p - 1$ and $m \le pq$,
  such that $\cI_{n,m}$ admits a set cover of size~$1$, whereas
  \Cref{long_alg:basic}, using~$p$ passes and running on $\cI_{n,m}$, returns
  a solution with $p(q-1) \approx pn^{1/p}$ sets.
\end{theorem}
\longonly{
\squeeze{\squeezelen}
\begin{proof}
  Put $\cX = [q]^p \setminus \{(q,q,\ldots,q)\}$. For each $j \in [p]$
  and $y \in [q-1]$, define the sets
  \begin{align*}
    \cX_j &= \{(x_1,\ldots,x_p) \in \cX:~
      x_1 = \cdots = x_{j-1} = q\} \,, \\
    S^y_j &= \{(x_1,\ldots,x_p) \in \cX:~
      x_1 = \cdots = x_{j-1} = q
      \,\wedge\, x_j = y\} \,.
  \end{align*}
  Then $\cX = \cX_1 \supseteq \cdots \supseteq \cX_p$ and $S^y_j
  \subseteq \cX_j$.  Observe that these sets~$S^y_j$ are pairwise
  disjoint and partition~$\cX$. Further, $|S^y_j| = q^{p-j}$ and
  $|\cX_j| = q^{p-j+1}-1$ for all $j,y$.

  Let~$\sigma_j$ be the stream consisting of the sets $\{S^y_j:\,
  y\in[q-1]\}$ in some arbitrary order. Let~$\sigma$ be the stream
  consisting of $\sigma_p$ followed by $\sigma_{p-1}$ and so on, down to
  $\sigma_1$, and finally the set~$\cX$. Consider the $\setcover_{n,m}$
  instance $\cI_{n,m}$ defined by $\sigma$: it satisfies $n = |\cX| =
  q^p-1$ and~$m=p(q-1)+1 \leq pq$, as claimed.  Since the entire
  universe~$\cX$ occurs as a set in~$\cI_{n,m}$, the optimum set cover
  consists of just that one set.

  Now consider the behaviour of \Cref{long_alg:basic} on~$\sigma$. For each
  $j\in[p]$, let $\tau_j = n^{1-j/p}$ be the threshold in the $j$th
  pass. We claim that
  \begin{equation} \label{long_eq:threshold-sandwich}
    q^{p-j}-1 < \tau_j \le q^{p-j} \,.
  \end{equation}
  The second inequality in~\eqref{long_eq:threshold-sandwich} is easy to see:
  $\tau_j = \left((q^p-1)^{1/p}\right)^{p-j} \le q^{p-j}$. The first
  inequality is obvious when $j = p$, so suppose that $1 \le j < p$.
  Consider the function
  \[
    G_{j,p}(x) = (x^p-1)^{1-j/p} - (x^{p-j}-1) \,.
  \]
  A routine calculation shows that the derivative $G'_{j,p}(x) =
  (p-j)x^{p-1}((x^p-1)^{-j/p} - x^{-j})$. For $x \ge 1$, we have
  $(x^p-1)^{1/p} < x$, so $(x^p-1)^{-j/p} > x^{-j}$; therefore
  $G'_{j,p}(x) > 0$. Since $G_{j,p}(1) = 0$, we now conclude that
  $G_{j,p}(q) > 0$, which gives us the first inequality
  in~\eqref{long_eq:threshold-sandwich} and proves the claim.
  We can now see that the $j^{\mathrm{th}}$ pass satisfies the following 
  properties.
  \begin{enumerate}[itemsep=1pt]
    \item At the start of the pass, the set of uncovered elements is
    precisely~$\cX_j$.
    \item Each set in $\sigma_p,\ldots,\sigma_{j+1}$ makes a
    contribution equal to its cardinality. Therefore the largest such
    contribution is $q^{p-j-1} < q^{p-j}-1 < \tau_j$,
    by~\eqref{long_eq:threshold-sandwich}.
    \item Each set in $\sigma_j$ makes a contribution equal to its
    cardinality, which is $q^{p-j} \ge \tau_j$,
    by~\eqref{long_eq:threshold-sandwich}.
    \item Each set in $\sigma_{j-1},\ldots,\sigma_1$ makes a
    contribution of zero.
    \item The set $\cX$, which arrives at the end of~$\sigma$, makes a
    contribution of $q^{p-j}-1 < \tau_j$,
    by~\eqref{long_eq:threshold-sandwich}.
    \item Therefore the sets added to $\Sol$ during the pass are exactly
    the sets in~$\sigma_j$.
  \end{enumerate}
  The validity of these properties can be formally proved by backward
  induction on~$j$. The details are routine and tedious, so we omit
  them.

  Based on the above properties, we see that \Cref{long_alg:basic} produces a
  solution consisting of all sets in all substreams~$\sigma_j$. The
  number of sets in this solution is $p(q-1)$, as claimed.
\end{proof}

% \tony{So, for $p=2$, Emek-Ros{\'e}n produces the same solution as us.
% They choose the largest threshold that leaves at most~$\sqrt n$ elements
% uncovered, each to be covered by its cheapest set. Assuming that it's
% first-come best dressed when sets are equally expensive, then they return the
% same solution as us. But under a different tie-breaking rule, they might in
% fact choose the set~$\cX$. Unclear whether to analyze this for higher~$p$.}
}

%!TEX root = main.tex

\section{The Basic Lower Bound} \label{long_sec:lb}

In this section we establish our main result, which gives a strong lower
bound on the best approximation ratio achievable by a semi-streaming
algorithm for \setcover. Our lower bound gives the optimal dependence of
this ratio on~$n$. Moreover, for~$p$ passes, our lower bound is only
about~$p^3$ times smaller than our upper bound in \Cref{long_thm:ub}. In
particular, when $p = \Theta(1)$, the lower bound is asymptotically
optimal.

\subsection{Warm Up: One-Pass Algorithms}

Our proof is based on a fairly technical combinatorial construction. To
motivate it, let us first outline a simple proof of a {\em one-pass}
lower bound. We start with the well-known \textsc{index} (or \idx)
problem in communication complexity, where Alice must send Bob a
(possibly random) message about her $n$-bit string~$x$, so that Bob, who
holds an index $h \in [n]$, can output~$x_h$ (the $h$th bit of~$x$) with
high probability. A textbook result~\cite{KushilevitzNisan-book} is that
this requires Alice to send $\Omega(n)$ bits. To reduce \idx to
\setcover, we construct a universe~$\cX$ and a family of~$n$ distinct
sets $S_1, \ldots, S_n \subset \cX$.  Alice encodes~$x$ as the stream of
sets $\{S_i:\, x_i = 1\}$, and Bob encodes~$h$ as a ``stream'' of just
one set: $\cX \setminus S_h$.  Alice's stream followed by Bob's is an
instance of \setcover.

When $x_h = 1$, this instance clearly has $|\Opt| = 2$. We can force
$|\Opt|$ to be much larger when $x_h = 0$ if we make each $|S_i|$ large
and each $|S_i \cap S_j|$ small (for $i \ne j$): since Alice's stream is
missing~$S_h$, it will take ``many'' sets~$S_i$, $i \ne h$, to cover the
elements of~$S_h$.

Incidence geometry gives us an elegant construction of a collection
$\{S_1,\ldots,S_n\}$ with these properties. Consider the lines of
an affine plane of order~$q$, with~$q$ a prime power. More explicitly,
let $\FF_q$ denote the finite field with~$q$ elements, $\cX = \FF_q^2$,
$n = |\cX| = q^2$, and $\{S_1,\ldots,S_n\}$ be some collection of~$n$
distinct lines
out of the $q^2 + q$ such lines in $\FF_q^2$.  Then each $|S_i| = q$ and
each $|S_i \cap S_j| \le 1$, for $i \ne j$.  In particular, $x_h = 0$
now implies that $|\Opt| \ge q = \sqrt n$.  Therefore approximating such a
\setcover instance to a factor smaller than $\sqrt n/2$ is enough to
solve \idx, whence an algorithm achieving such approximation
must use~$\Omega(n)$ space. 

To rule out a semi-streaming algorithm we must prove a stronger,
$n^{1+\Omega(1)}$ space, lower bound. A simple tweak achieves this:
sticking with the universe $\FF_q^2$, replace the lines in the above
construction with degree-$2$ algebraic curves, say. This preserves the
essential dichotomy between large~$|S_i|$ and small $|S_i \cap S_j|$
while allowing us to reduce from an \idx instance on $n^{1+\Omega(1)}$
bits.

The one-pass lower bound proof we have just outlined is arguably more
straightforward than the Emek--Ros\'en proof~\cite{EmekR14}. 
Though both
proofs begin with the affine plane,
our builds an explicit set system, rather than relying on
a probabilistic argument,
and reduces directly from \idx, rather than a employing bespoke
entropy calculations, leading to a more modular proof.
%theirs goes on to build a more
%complicated set system using the probabilistic method and then employs
%ad hoc entropic calculations (rather than a straightforward reduction
%from \idx) to deduce the space lower bound. 
%\mnote{Consider rewording this to say {\em our} set system is explict,
%avoids probabilistic method, and our reduction from IDX leads to modular
%proof.}
%
But there is a far more
important takeaway from our proof: the observation that employing
higher-degree curves adds great flexibility to the construction.
Exploiting this observation to its fullest allows us to handle
multi-pass algorithms by greatly generalising the construction, moving
from affine planes to more abstract incidence geometries that we call
{\em edifices} (\Cref{long_def:edifice} below). Edifices, like affine planes,
are examples of Buekenhout geometries~\cite{Buekenhout79}.

\subsection{Multi-Player Tree Pointer Jumping}

A popular source problem for multi-pass streaming lower bounds is the
communication problem {\em multi-player tree pointer jumping}, which
generalises \idx. Let $T$ be a rooted tree with $k\ge 2$ layers of
vertices, where a vertex is in layer~$\ell$ if it is at distance exactly
$k - \ell$ from the root (thus, the root is in layer~$k$) and every leaf
is in layer~$1$. The pointer jumping problem on~$T$, denoted $\mpj_T$,
is a $k$-player number-in-hand communication game involving players
named $\plr_1,\ldots,\plr_k$. For $1 < j \le k$, $\plr_j$'s input
specifies one {\em pointer} (i.e., out-edge) at each vertex in layer~$j$;
by definition each such pointer leads to a vertex in layer $j-1$.
Furthermore, $\plr_1$'s input specifies a bit at each layer-$1$ vertex;
these bits are called {\em leaf bits}.
Given such an input, $\pi$, let
$T|_\pi$ denote the subgraph of~$T$ defined by retaining only those
edges of~$T$ that correspond to pointers in~$\pi$.
Then~$T|_\pi$
contains a unique root-to-leaf path, ending at a leaf $v_\pi$, say.
The desired output corresponding to~$\pi$, denoted $\mpj_T(\pi)$, is
defined to be the leaf bit at~$v_\pi$. 

The communication game involves players announcing messages on a shared
broadcast channel, according to a public-coin randomised protocol. The
protocol proceeds in {\em rounds}, where a round is defined as one
message each from $\plr_1,\ldots,\plr_k$, speaking in that order. The
last message of the protocol must be a single bit, which is defined to
be the protocol's output. An $[r,C,\eps]$-protocol for $\mpj_T$ is
defined to be one in which
\shortonly{
(i) there are at most~$r$ rounds of communication;
(ii) within each round, the total number of bits communicated is at most~$C$; and
(iii) with probability at least $1-\eps$, the protocol's output equals
$\mpj_T(\pi)$.
}
\longonly{
\begin{itemize}[itemsep=1pt]
\item there are at most~$r$ rounds of communication;
\item within each round, the total number of bits communicated is at
most~$C$; and
\item the protocol's output equals $\mpj_T(\pi)$ with probability at
least $1-\eps$.
\end{itemize}
}

\begin{definition}
  The $r$-round randomised communication complexity of $\mpj_T$ is
  defined to be $\R^r(\mpj_T) := \min\{C:$ there exists an
  $[r,C,\frac13]$-protocol for $\mpj_T\}$.
\end{definition}

% \shortonly{
% We shall use the following lower bound, due to 
% Chakrabarti, Cormode, and McGregor~\cite{ChakrabartiCM08}.
% }
%
\longonly{
Intuitively, if players trying to solve $\mpj_T$ are restricted to a
``small'' amount of communication per round, then because they are
forced to speak in the ``wrong'' order, in the first round the only
player who is able to convey ``useful'' information is $\plr_k$, in the
second round the only such player is $\plr_{k-1}$, and so on.
Therefore, if the protocol is further restricted to $k-1$ rounds,
$\plr_1$ rarely gets a chance to convey useful information and so the
protocol's error probability should be high. This intuition was
formalised in the {\em round elimination} ideas of Miltersen
\etal~\cite{MiltersenNSW98}. Using these ideas and a {\em direct sum}
argument, Chakrabarti, Cormode, and McGregor~\cite{ChakrabartiCM08}
proved a distributional communication complexity lower bound for \mpj.
We only need the consequent randomised communication complexity bound,
stated below.
}

\shortonly{\begin{theorem}[{Chakrabarti, Cormode, and McGregor~\cite[Theorem 4.5]{ChakrabartiCM08}}]}
\longonly{\begin{theorem}[{\cite[Theorem 4.5]{ChakrabartiCM08}}]}
\label{long_thm:mpj}
  Let~$T$ be a complete $t$-ary tree with $k \ge 2$ layers of vertices.
  Then $\R^{k-1}(\mpj_T) = \Omega(t/k^2)$. \qed
\end{theorem}

\subsection{Reduction to Set Cover via Edifices}

\begin{definition} \label{long_def:edifice}
A {\em $(k,d,q,t)$-edifice} $\cT$ over a universe~$\cX$ is a rooted tree,
together with an associated collection of sets called the {\em
varieties} of the edifice, satisfying the following properties.

\shortonly{
\begin{enumerate}[topsep=2pt,itemsep=0pt,label=(E\arabic*),leftmargin=*,widest={(E5)}]
}
\longonly{
\begin{enumerate}[itemsep=1pt,label=(E\arabic*),leftmargin=*,widest={(E5)}]
}
\item \label{long_ax:arity} $\cT$ is a complete $t$-ary tree, i.e., every
non-leaf vertex has exactly~$t$ children.
\item \label{long_ax:depth} $\cT$ has~$k$ levels (equivalently, depth $k-1$),
numbered~$1$ through~$k$ from leaves to root.
\item \label{long_ax:variety} Each vertex~$v$ of~$\cT$ has an associated set
$X_v \subseteq \cX$, called the {\em variety} at~$v$.
\item \label{long_ax:incidence} If~$u$ is the parent of~$v$, then~$X_u
\supseteq X_v$. If~$r$ is the root of~$\cT$, then $X_r = \cX$.
\item \label{long_ax:sink} If~$z$ is a leaf of~$\cT$, then $|X_z| \ge q$.
\item \label{long_ax:intersection} For each leaf~$z$ of~$\cT$ and each
vertex~$v$ not an ancestor of~$z$, we have $|X_z \cap X_v| \le d+k-1$.
\end{enumerate}
\end{definition}

The $(k,d,q,t)$-edifices that interest us will have $k \approx d \ll q
\ll t$.  In particular, if $d+k \le q$ and $t \ge 2$, it is easy to
prove from \ref{long_ax:incidence}, \ref{long_ax:sink}, and \ref{long_ax:intersection}
that varieties at distinct vertices are distinct (as sets).  For readers
familiar with incidence geometry, we note that these varieties then form
a Buekenhout geometry~\cite{Buekenhout79} of rank~$k$, where the type
map sends each variety to the level of its corresponding vertex and the
incidence relation is symmetrised set inclusion. Thus our notion of an
edifice generalises affine planes, which we used in our warm-up proof:
an affine plane over $\FF_q$ is a $(2,0,q,q^2+q)$-edifice over the
universe $\FF_q^2$.

\begin{theorem} \label{long_thm:reduction}
  Suppose there exists a $(p+1,d,q,t)$-edifice $\cT$ with $q >
  (p+d)(p+1)$.  Then every randomised $p$-pass streaming algorithm that,
  with probability at least $2/3$, approximates $\setcover$ to a factor
  smaller than $q/\left((p+d)(p+1)\right)$ must use at least $\R^p(\mpj_\cT)/p =
  \Omega(t/p^3)$ bits of space.
\end{theorem}
\squeeze{.7\squeezelen}
\begin{proof}
  Let the edifice~$\cT$ be over a universe~$\cX$.  We shall transform an
  input~$\pi$ to $\mpj_\cT$ into an instance~$\cI(\pi)$ of \setcover on
  the universe~$\cX$, with each set in~$\cI(\pi)$ being assigned to one
  of $\plr_1,\ldots,\plr_{p+1}$.

  The transformation is as follows. Let~$u$ be a vertex of~$\cT$ in
  layer $j \ge 2$. Then~$\pi$ specifies a pointer from~$u$ to some
  vertex, say~$v$.  We encode this pointer as the set $X_u \setminus
  X_v$ and assign this set to $\plr_j$. 
  \longonly{
  We perform this encoding for each vertex in layers~$2$ and higher.
  }
  Furthermore, we encode the leaf
  bits of~$\pi$ as the collection of sets $\{X_z:\, \pi$ specifies a
  `$1$' at leaf $z\}$ and assign all sets in this collection to
  $\plr_1$.  Finally, we assign every singleton subset of~$\cX$ to
  $\plr_1$.  This completes the specification of our \setcover instance,
  which is valid thanks to the inclusion of the singletons. 

  Let $v_{p+1},\ldots,v_1$ be the unique root-to-leaf path in
  $\cT|_\pi$, with~$v_j$ being in layer~$j$, for each $j$. Put $X_j =
  X_{v_j}$, for each~$j$. By \ref{long_ax:incidence}, $\cX = X_{p+1}
  \supseteq \cdots \supseteq X_1$, so the encodings of the pointers at
  $v_{p+1},\ldots,v_2$ together cover $\bigcup_{j=2}^{p+1} (X_j
  \setminus X_{j-1}) = \cX \setminus X_1$.  Now suppose that
  $\mpj_\cT(\pi) = 1$. Then the encoding of the leaf bits
  includes~$X_1$, so $\cI(\pi)$ has a set cover of size $Q_1 := p+1$. 

  Next, suppose that $\mpj_\cT(\pi) = 0$. A set cover must, in
  particular, cover~$X_1$.  However, the encodings of the pointers at
  $v_{p+1},\ldots,v_2$ are all disjoint from~$X_1$ and the encoding of
  the leaf bits does not include~$X_1$. Therefore, $X_1$ must be covered
  using only singletons and sets corresponding to non-ancestors
  of~$v_1$.  For each such non-ancestor, $y$, the corresponding set in
  $\cI(\pi)$ is a subset of the variety~$X_y$.  By
  \ref{long_ax:intersection}, such a set covers at most $d+p$ elements of
  $X_1$ whereas, by \ref{long_ax:sink}, $|X_1| \ge q$.  Therefore every set
  cover in $\cI(\pi)$ uses least $Q_0 := q/(d+p)$ sets.

  It follows that approximating even the optimum {\em value} of
  $\cI(\pi)$ to a factor smaller than $Q_0 / Q_1 = q/((p+d)(p+1))$ is
  sufficient to determine $\mpj_\cT(\pi)$.
  \shortonly{
  Reasoning along standard lines, a $p$-pass, $s$-space-bounded,
  $\frac13$-error randomised streaming algorithm that approximates
  \setcover this well yields a $[p,sp,\frac13]$-protocol for
  $\mpj_\cT(\pi)$, whence $sp \ge \R^p(\mpj_\cT)$.
  }
  \longonly{

  Let~$\cA$ be a $p$-pass $\frac13$-error randomised streaming algorithm
  that approximates \setcover this well, using at most~$s$ bits of space.
  The players can solve $\mpj_\cT$ as follows. On input~$\pi$, each
  player follows the above encoding scheme so that players jointly
  arrive at the \setcover instance $\cI(\pi)$, with sets assigned
  amongst the players. They simulate the execution of~$\cA$
  on the stream~$\sigma$ obtained by taking $\plr_1$'s sets, followed by
  $\plr_2$'s sets, and so on. Each time the execution of~$\cA$ moves off
  one player's portion of~$\sigma$, that player broadcasts the memory
  contents of~$\sigma$.  This simulation uses one communication round
  per streaming pass, and spends~$sp$ bits of communication per round.
  Therefore it yields a $[p,sp,\frac13]$-protocol for $\mpj_\cT(\pi)$,
  whence $sp \ge \R^p(\mpj_\cT)$.
  }
\end{proof}

\subsection{Construction of an Edifice}

\begin{theorem} \label{long_thm:edifice}
  Let $k,d$, and~$q$ be integers with $k\ge 1$, $d \ge 0$, and $q \ge
  d+k$, with~$q$ being a prime power. Then there exists a
  $(k,d,q,q^{d+k}(1-1/q))$-edifice.
\end{theorem}
\squeeze{.7\squeezelen}
\begin{proof}
  We shall construct an explicit edifice over the universe $\cX =
  \FF_q^k$. The varieties of our edifice will be certain well-structured
  varieties in the sense of algebraic geometry, i.e., solution sets of
  polynomial equations.
  Write the coordinates of a generic point in $\FF_q^k$ as
  $(x,y_1,\ldots,y_{k-1})$. An {\em edificial equation} of {\em rank}~$i$ is defined
  to be an equation of the form
  \begin{equation} \label{long_eq:edificial}
    y_i = \ell_i(y_1,\ldots,y_{i-1},f_{k-i}(x)) \,,
    \qquad 1 \le i \le k-1 \,,
  \end{equation}
  where $\ell_i(z_1,\ldots,z_i)$ is a homogeneous linear form over
  $\FF_q$ whose $z_i$-coefficient is nonzero and~$f_j(x)$ is a monic
  polynomial in $\FF_q[x]$ of degree exactly $d+j$. \Cref{long_eq:edificial}
  is abbreviated as $\ed{\ell_i}{f_{k-i}}$. 

  Notice that irrespective of the value of~$i$ there are exactly $d+k$
  coefficients appearing on the right-hand side of \cref{long_eq:edificial},
  one of which must be nonzero. There are exactly $t := q^{d+k}(1-1/q)$
  ways to choose these coefficients, leading to exactly~$t$ distinct
  edificial equations of each rank.

  Let~$\cT$ be a rooted complete $t$-ary tree with~$k$ levels, the
  root~$r$ being at level~$k$. For $1 \le i \le k-1$, for each
  level-$(i+1)$ vertex~$v$ of~$\cT$, label each of the~$t$ edges leaving
  $v$ with one of the~$t$ distinct rank-$i$ edificial equations.
  Associate a variety~$X_v$ with vertex~$v$ as follows.
  Let $X_r = \cX$. If $v \ne r$, let~$X_v$ be the variety
  defined by the set of edificial equations labelling the edges on the
  path from~$r$ to~$v$. We shall show that~$\cT$, with these associated
  varieties, forms a $(k,d,q,t)$-edifice. Certainly, properties
  \ref{long_ax:arity}, \ref{long_ax:depth}, \ref{long_ax:variety}, and
  \ref{long_ax:incidence} are immediate. The following observation will be
  helpful in establishing the remaining properties.

  \begin{observation} \label{long_obs:edificial}
    Suppose $\bx = (x,y_1,\ldots,y_{k-1})$ satisfies the edificial equations
    $\ed{\ell_1}{f_{k-1}},\ldots,\ed{\ell_j}{f_{k-j}}$ for some $j$ with $1
    \le j \le k-1$. Then there exist linear forms $\lambda_i(z_1,\ldots,z_i)$
    over $\FF_q$ such that
    \begin{equation} \label{long_eq:determined}
      y_i = \lambda_i(f_{k-1}(x),\ldots,f_{k-i}(x)) \,,
      \qquad 1 \le i \le j \,.
    \end{equation}
    Therefore each of $y_1,\ldots,y_j$ is determined by~$x$.
  \end{observation}

  For the rest of this proof let~$z$ be a leaf; let $\bx =
  (x,y_1,\ldots,y_{k-1}) \in X_z$ be an arbitrary point in the variety
  at~$z$ and let $\ed{\ell_1}{f_{k-1}},\ldots,\ed{\ell_{k-1}}{f_1}$ be
  the edificial equations defining $X_z$. 
  \shortonly{%
  We record a corollary.
  }%
  \longonly{%
  We record the following corollary of \Cref{long_obs:edificial}.
  }

  \begin{observation} \label{long_obs:determined}
    The point~$\bx$ is completely determined by its first coordinate $x$.
  \end{observation}

  It follows that for each $a \in \FF_q$, $X_z$ contains exactly one such point
  $\bx$ with $x = a$, whence $|X_z| = |\FF_q| = q$. This establishes property
  \ref{long_ax:sink}.

  Property~\ref{long_ax:intersection} requires a more careful examination of the
  form of the edificial equations. Consider a vertex~$v$ that is not an
  ancestor of the leaf~$z$. Let~$u$ be the highest (by level) ancestor of~$v$
  that is still not an ancestor of~$z$. Since $X_u \supseteq X_v$, it
  suffices to prove that $|X_z \cap X_u| \le d+k-1$. Suppose~$u$ is at level $j
  < k$.  Then~$X_u$ is defined by the $k-j-1$ highest-ranked edificial
  equations that define~$X_z$ (which are of ranks $k-1$ through $j+1$) plus an
  additional rank-$j$ equation $\ed{\ell^+_j}{f^+_{k-j}}$, where either
  $\ell_j \ne \ell^+_j$ or $f_{k-j} \ne f^+_{k-j}$, or both.

  Suppose that $\ell_j = \ell^+_j$, so that $f_{k-j} \ne f^+_{k-j}$.  Each
  point $\bx = (x,y_1,\ldots,y_{k-1}) \in X_z \cap X_u$ must, in particular
  satisfy $\ed{\ell_j}{f_{k-j}}$ and $\ed{\ell_j}{f^+_{k-j}}$.  Comparing
  these two equations gives
  \begin{align}
    & \ell_j(y_1,\ldots,y_{j-1},f_{k-j}(x)) = y_j 
      = \ell_j(y_1,\ldots,y_{j-1},f^+_{k-j}(x)) \label{long_eq:compare} \\
    \Rightarrow\quad
    & \ell_j(0,\ldots,0,f_{k-j}(x)-f^+_{k-j}(x)) = 0 \notag \\
    \Rightarrow\quad
    & f_{k-j}(x)-f^+_{k-j}(x) = 0 \label{long_eq:diffzero} \,,
  \end{align}
  because the linear form $\ell_j(z_1,\ldots,z_j)$ is required to have a
  nonzero $z_j$-coefficient. The left-hand side of \cref{long_eq:diffzero} is
  a nonzero univariate polynomial of degree at most $d+k-j$, whence it
  has at most $d+k-j$ roots in $\FF_q$. By \Cref{long_obs:determined}, it
  follows that $|X_z \cap X_u| \le d+k-j \le d+k-1$.

  Finally, suppose $\ell_j \ne \ell^+_j$. We now make the crucial observation
  that
  \begin{equation} \label{long_eq:indep}
    f_1(x), \ldots, f_{k-1}(x)~ \text{are linearly independent over}~ \FF_q \,,
  \end{equation}
  which holds because these polynomials have distinct degrees. With this in
  mind, examining \cref{long_eq:edificial,long_eq:determined} and recalling that
  $\ell_i(z_1,\ldots,z_i)$ has a nonzero $z_i$-coefficient, we see that
  $\lambda_i(z_1,\ldots,z_i)$ also has a nonzero $z_i$-coefficient.
%\tony{This doesn't seem obvious to me yet. Amit: Though technical, this proof
%follows very well, except (for me anyway) this point. Could you add a sentence
%or two to help out?}
  Therefore,
  for each $i \in \{1,\ldots,k-1\}$, 
  the collection of polynomials
  $\{\lambda_1(f_{k-1}(x)),\, \ldots,\, \lambda_i(f_{k-1}(x),\ldots,f_{k-i}(x))\}$
  is a basis for the linear subspace of $\FF_q[x]$ spanned by 
  $\{f_{k-1}(x),\ldots,f_{k-i}(x)\}$.

  Suppose $\bx = (x,y_1,\ldots,y_{k-1}) \in X_z \cap X_u$.  Proceeding as in
  \cref{long_eq:compare}, we find that
  \[
    \ell_j(y_1,\ldots,y_{j-1},f_{k-j}(x)) = y_j 
    = \ell^+_j(y_1,\ldots,y_{j-1},f^+_{k-j}(x)) \,.
  \]
  Therefore there exists a linear form $h(z_1,\ldots,z_{j-1})$ and scalars
  $a,a^+\in\FF_q$, where either $h \ne 0$ or $a \ne a^+$ or both, such that
  $h(y_1,\ldots,y_{j-1}) + a f_{k-j}(x) - a^+ f^+_{k-j}(x) = 0$.
  By \Cref{long_obs:edificial},
  \begin{align}
    h(\lambda_1(f_{k-1}(x)),\,
    \ldots,\, \lambda_{j-1}(f_{k-1}(x),\ldots,f_{k-j+1}(x)))
    + a f_{k-j}(x) - a^+ f^+_{k-j}(x) = 0 \,. \label{long_eq:lincomb}
  \end{align}
  We claim that the left-hand side of \cref{long_eq:lincomb} is a nonzero
  polynomial. If $h = 0$, this is immediate because $a \ne a^+$, whereas
  $f_{k-j}(x)$ and $f^+_{k-j}(x)$ are both monic of degree $d+k-j$.  If
  $h \ne 0$, then by our observations about the polynomials
  $\{\lambda_i(f_{k-1}(x),\ldots,f_{k-i}(x))\}$, the first term on the
  left-hand side is a nonzero polynomial in the span of
  $\{f_{k-1}(x),\ldots,f_{k-j+1}(x)\}$. In particular, its degree is at
  least $d+k-j+1$. The other two terms have degree at most $d+k-j$,
  which proves the claim.

  Thus, \cref{long_eq:lincomb} states that~$x$ is a root of a nonzero
  polynomial of degree at most $d+k-1$, a fact we derived from the
  condition that $\bx \in X_z \cap X_u$. By \Cref{long_obs:determined}, $|X_z
  \cap X_u| \le d+k-1$.
\end{proof}

\longonly{
\paragraph{Justifications for observations.}
For the sake of completeness, we formally justify the observations made
in the course of the just-concluded proof. \Cref{long_obs:edificial} can be
proved by induction on $i$. When $i = 1$, \cref{long_eq:edificial}
specialises to $y_1 = \ell_1(f_{k-1}(x))$, so we reach
\cref{long_eq:determined} by taking $\lambda_1 = \ell_1$. For general $i$, by
the induction hypothesis, we have
\begin{equation} \label{long_eq:unfold}
  y_i = \ell_i(\lambda_1(f_{k-1}(x)), \ldots,
  \lambda_{i-1}(f_{k-1}(x),\ldots,f_{k-i+1}(x)), f_{k-i}(x)) \,.
\end{equation}
Each argument to $\ell_i$ in the above equation is a linear form in
$\{f_{k-1}(x), \ldots, f_{k-i}(x)\}$, and $\ell_i$ is itself a linear
form. Taking $\lambda_i$ to be the ``composition'' of these linear
forms gives us \cref{long_eq:determined}.

\Cref{long_obs:determined} is, as noted, a simple corollary to \Cref{long_obs:edificial}.

We turn to the observation, made just after~\eqref{long_eq:indep}, that
$\lambda_i(z_1,\ldots,z_i)$ has a nonzero $z_i$-coefficient.  Of the $i$
arguments to $\ell_i$, only the last involves $f_{k-i}(x)$, and that
last argument is given a nonzero coefficient by the defining property of
$\ell_i$. The other arguments are polynomials in the span of
$\{f_{k-1}(x),\ldots,f_{k-i+1}(x)\}$. The linear independence observed
in~\eqref{long_eq:indep} completes the justification.
}

\subsection{Pass/Approximation Tradeoff for Set Cover}

We now bring together our technical results to obtain a pass/approximation
tradeoff for \setcover in the semi-streaming setting. 

\begin{theorem}[Main result] \label{long_thm:lb}
  Let $c > 1$ be a constant. Let~$\cA$ be a $p$-pass streaming algorithm
  that, for all large enough~$n$ and~$m$, approximates the optimum value
  of $\setcover_{n,m}$ instances to a factor smaller than
  $n^{1/(p+1)}/(c(p+1)^2)$ with probability at least $2/3$. Then~$\cA$
  must use $\Omega(n^c/p^3)$ bits of space. This space lower bound
  applies to instances with $m = \Theta(n^{cp})$.
\end{theorem}
\squeeze{.7\squeezelen}
\begin{proof}
  Let~$q$ be a sufficiently large prime power. Put $d = (c-1)(p+1)$, $t
  = q^{d+p+1}(1-1/q)$, and $n = q^{p+1}$. By \Cref{long_thm:edifice}, there
  exists a $(p+1,d,q,t)$-edifice over a universe $\cX$ with $|\cX| = n$.
  By \Cref{long_thm:reduction}, the space usage of~$\cA$, which approximates
  \setcover to a factor better than $n^{1/(p+1)}/((p+d)(p+1))$, is at
  least $\R^p(\mpj_T)/p$, where~$T$ is a complete $(p+1)$-level $t$-ary
  tree. By \Cref{long_thm:mpj}, this space bound is $\Omega(t/p^3) =
  \Omega(n^{d/(p+1)+1}(1-1/q)/p^3) = \Omega(n^c/p^3)$.

  Examining the reduction in \Cref{long_thm:reduction} shows that instances
  of \setcover demonstrating the above lower bound have roughly as many
  sets as the edifice has leaves, i.e., $m = \Theta(n^{cp})$.
\end{proof}

\noindent
It is instructive to note the following corollaries of \Cref{long_thm:lb}.
\begin{enumerate}[itemsep=1pt]
\item Let~$p$ be a constant. Then there exist positive constants $\alpha
< 1$ and $\beta > 1$ such that $(\alpha n^{1/(p+1)})$-approximating
\setcover in $p$ streaming passes requires $\Omega(n^\beta)$ space. In
particular, such an approximation is not possible for a semi-streaming
algorithm.
\item Every multi-pass semi-streaming $O(\log n)$-approximation
algorithm for \setcover requires $p = \Omega(\log n/\log\log n)$ passes.
\end{enumerate}

\longonly{
\subsection{Two-Player Communication Complexity of Set Cover}

Nisan~\cite{Nisan02} and Demaine \etal~\cite{DemaineIMV14} have studied
\setcover as a communication game. Our proof of \Cref{long_thm:lb} directly
implies a lower bound for a certain {\em multi}-player \setcover game.
But one may wonder about implications for the more fundamental setting
of {\em two}-player communication complexity. Our next theorem shows
that our technology does indeed yield a new two-player result.

In the two-player \setcover game, there is a fixed finite universe $\cX
= [n]$, Alice receives as input a collection $\cF \subseteq 2^\cX$, and
Bob receives a collection $\cG \subseteq 2^\cX$. The players wish to
solve the \setcover instance $(\cX, \cF \cup \cG)$ as cheaply as
possible. Specifically, they must output a cover certificate (analogous
to the array $\Cov$ in \Cref{long_alg:basic}) that specifies, for each
$x\in\cX$, the set in $\cF \cup \cG$ that covers $x$. A communication
protocol that gives such an output is said to be $\alpha$-approximate if
the implied set cover, $\Sol$, satisfies $|\Sol| \le \alpha|\Opt|$,
where $\Opt$ is an optimum solution to the instance.

By mimicking the standard offline greedy algorithm for \setcover, one
readily obtains a $(\ln n - \ln\ln n + \Theta(1))$-approximate protocol
that communicates at most $n$ messages, each message being $n$ bits
long; in particular, the total communication cost is $n^{O(1)}$.  Nisan
proved~\cite[Theorem~4]{Nisan02} that for every constant $\delta > 0$, a
$(\frac12-\delta)\log_2 n$-approximate protocol requires an amount of
communication that is exponentially larger, roughly $\exp(\sqrt n)$ for
small $\delta$. Nisan's theorem uses a reduction from
\textsc{set-disjointness} and is therefore agnostic about the number of
messages in the protocol. Our theorem complements this by giving a
``bounded-round'' lower bound.

\begin{theorem} \label{long_thm:comm}
  Let $c > 1$ be a constant.  Suppose there exists a (randomised)
  $\alpha$-approximate protocol for the two-player \setcover game that
  communicates a total of $C$ bits in at most $r$ messages. Then either
  $\alpha \ge n^{1/(r+1)}/(c(r+1)^2)$ or $C = \Omega(n^c/r^2)$.
\end{theorem}
\begin{proof}[Proof sketch]
  We encode an instance of \textsc{pointer-jumping} on a tree as a
  \setcover instance, using our edifices, exactly as in the proof of
  \Cref{long_thm:reduction}. We then treat \textsc{pointer-jumping} as
  a {\em two}-player communication game, with Alice holding the
  information at vertices of the tree whose level is odd, and Bob
  holding the rest. For this two-player game, we invoke the
  bounded-round communication lower bound due to Klauck
  \etal~\cite{KlauckNTZ01} to finish the proof.
\end{proof}

While we could have used the above two-player version of
\textsc{pointer-jumping} as the basis for a data-streaming lower bound,
it is important to note that doing so would have considerably weakened
the streaming result, because $p$ streaming passes translate into $2p-1$
messages in a two-player protocol.
}

%!TEX root = main.tex

\section{Extension to Partial Cover}

\shortonly{
As foreshadowed, there is a relaxation of \setcover in which
a feasible solution may leave ``a few'' elements uncovered.
}%
\longonly{
Thus far we have focused on the \setcover problem as traditionally
defined, in which a feasible solution must cover the entire universe.
However, as is the case with many optimisation problems, \setcover
admits a relaxation in the form of a {\em bicriterial approximation},
wherein the feasibility constraint can be violated by some amount~$\eps$,
and we seek a solution with cost at most $\alpha(\eps,n)$ times
the optimum fully feasible solution, for some function~$\alpha$. 

}%
To be precise, we consider the problem $\partcover_{n,m,\eps}$, where an
instance $\cI = (\cX,\cF,\eps)$ consists of a universe $\cX$, with $|\cX|
= n$, a collection of sets $\cF \subseteq 2^\cX$ with $|\cF| = m$, and a
parameter $\eps \in [0,1]$. The goal is to compute a $(1-\eps)$-partial
cover of~$\cX$, defined as a collection $\Sol \subseteq \cF$ that covers
at least $(1-\eps)|\cX|$ elements. Such a solution $\Sol$ is said to be
$\alpha$-approximate if $|\Sol| \le \alpha |\Opt|$---or, in the weighted
version, $w(\Sol) \le \alpha w(\Opt)$---where $\Opt$ is a
minimum-cost set cover for $(\cX,\cF)$. Notice that we are comparing the
cost of our partial cover with that of the best \emph{total} cover.

\subsection{Upper Bound}
\label{long_sec:partial-upper-bound}

\shortonly{
Our upper bound generalises \Cref{long_thm:ub} (up to constant ``$8$''),
by including partial covers and weighted sets. A proof appears in
\Cref{long_sec:partial-upper-bound}.  For succinctness,
we assume all weights are $O(\log m)$-bit integers. 
}
\longonly{
We begin with our most general upper bound, which includes partial
covers and weighted sets. For convenience in stating the space bound, we
assume that all weights are $O(\log m)$-bit integers. 
}

\begin{theorem} \label{long_thm:ub-partial}
  For every integer $p \ge 1$, there is a $p$-pass, $O(n\log m)$-space
  algorithm for the {\em weighted} version of $\partcover_{n,m,\eps}$
  that produces an $\alpha(n,\eps)$-approximate cost $(1-\eps)$-partial
  cover, where
  %\[
    $\alpha(n,\eps) = \min\{ 8p\eps^{-1/p}, (8p+1)n^{1/(p+1)} \}$.
  %\]
\end{theorem}

\longonly{
\squeeze{\squeezelen}
\begin{proof}
  We run the following two schemes in parallel, returning the lower-cost
  solution.  First, we run the Emek--Ros\'{e}n algorithm for~$p$ passes,
  each time obtaining a $(1-\eps^{1/p})$-partial cover of the remaining
  (uncovered) portion of~$\cX$, and each time adding at most
  $8\eps^{-1/p} w(\Opt)$ cost to our solution~$\Sol$.  By definition of
  a partial cover, for each $j \in [p]$, the collection of sets
  constituting $\Sol$ after~$j$ passes leaves at most $\eps^{j/p} |\cX|$
  elements uncovered. Therefore, in the end, $\Sol$ is a
  $(1-\eps)$-partial cover.

  Second, we run the Emek--Ros{\'e}n algorithm for~$p$ passes (again) but here, in each pass,
  obtaining a~$(1-1/n^{1/(p+1)})$-partial cover of the remaining
  (uncovered) portion of~$\cX$, and each time adding at most
  $8n^{1/(p+1)} w(\Opt)$ cost to our solution~$\Sol$.  The collection of
  sets constituting $\Sol$ after~$j$ passes leaves at most
  $n^{(p+1-j)/(p+1)}$ elements uncovered.  After~$p$ passes, $\Sol$
  covers all but at most $n^{1/(p+1)}$ elements.  Covering each of these
  with its cheapest-covering set---which the Emek--Ros{\'e}n algorithm
  records---of cost at most $w(\Opt)$, leads to a total cost of at most
  $(8p+1)n^{1/(p+1)}w(\Opt)$.  Since this is a full cover of~$\cX$, it
  is also a $(1-\eps)$-partial cover.
\end{proof}

The above upper bound generalises \Cref{long_thm:ub}, except for the constant
``$8$'' that arises in the Emek--Ros\'en analysis. One can tweak their
algorithm so as to replace the~$8$ with $(1+\delta)^3$, where $\delta >
0$ is a constant of our choice, at the cost of increasing the space
usage by a factor of $\Theta(1/\log(1+\delta))$, which is about
$\Theta(\delta^{-1})$ for small~$\delta$.
}

%\begin{theorem} \label{long_thm:ub-partial}
%  For every integer $p \ge 1$, there is a $p$-pass semi-streaming
%  algorithm for $\partcover_{n,m,\eps}$ that produces an
%  $\alpha(n,\eps)$-approximate $(1-\eps)$-partial cover, where
%  \[
%    \alpha(n,\eps) = \min\{pn^{1/(p+1)}, O(p\eps^{-1/p})\} \,.
%  \]
%\end{theorem}
%\begin{proof}
%  When~$\eps$ is small enough, we use the algorithm in \Cref{long_sec:folded} to achieve an
%  approximation ratio of $pn^{1/(p+1)}$. Otherwise, we run the
%Emek-Ros\'{e}n algorithm
%  for~$p$ passes, each time obtaining a $(1-\eps^{1/p})$-partial cover
%  of the remaining (uncovered) portion of~$\cX$, and each time adding at most
%  $O(\eps^{-1/p}) |\Opt|$ sets
%%, which we accumulate over all the passes
%%  to constitute our
%to our solution~$\Sol$.
%
%  By definition of a partial cover, for each $j \in [p]$, the collection
%  of sets constituting $\Sol$ after~$j$ passes leaves at most $\eps^{j/p}
%  |\cX|$ elements uncovered. Therefore, in the end, $\Sol$ is a
%  $(1-\eps)$-partial cover.
%\end{proof}

\subsection{Lower Bound}
\label{long_sec:partial-lower-bound}

We shall now show that \Cref{long_thm:ub-partial} is asymptotically tight for
every constant $p$ by proving an appropriate lower bound on the
approximation factor that a semi-streaming algorithm for \partcover can
achieve. Our lower bound will hold even for unweighted \partcover and
will match the upper bound of \Cref{long_thm:ub-partial} up to
a~$\Theta(p^3)$ factor.

Our proof is based on edifices---as in the proof of
\Cref{long_thm:lb}---except that we need a different, more complicated,
setting of parameters that is not directly achieved by
\Cref{long_thm:edifice}. Instead, we revisit the edifices constructed in the
proof of \Cref{long_thm:edifice} and observe that they have an additional
geometric property that we call {\em wideness}: roughly speaking, each
level contains many groups of mutually parallel varieties.  Clustering
these parallel classes into ``supervarieties'' gives us new edifices
with the desired parameters.

Let $\cC_\cT(u)$ denote the set of
children of a vertex~$u$ in a tree~$\cT$.
A $(k,d,q,t)$-edifice $\cT$ is said to be {\em $(b,t')$-wide} if, for
each non-leaf vertex~$u$ of~$\cT$, there exist subsets
$\cV_1,\ldots,\cV_{t'} \subseteq \cC_\cT(u)$ such that
\shortonly{
(i) $\cV_1,\ldots,\cV_{t'}$ are pairwise disjoint; 
(ii) for all $i\in[t']$, $|\cV_i| = b$; and 
(iii) for all $i\in[t']$, for all $v\ne v'\in
\cV_i$, we have $X_v \cap X_{v'} = \varnothing$.
}
\longonly{
\begin{enumerate}[itemsep=1pt,label=(W\arabic*),leftmargin=*,widest={(W5)}]
\item \label{long_ax:partition} $\cV_1,\ldots,\cV_{t'}$ are pairwise disjoint;
\item \label{long_ax:width} for all $i\in[t']$, $|\cV_i| = b$; and
\item \label{long_ax:parallel} for all $i\in[t']$, for all $v\ne v'\in
\cV_i$, we have $X_v \cap X_{v'} = \varnothing$.
\end{enumerate}
}

\shortonly{
Complete proofs of the following lemmas and theorem appear in
\Cref{long_sec:partial-lower-bound}.}
\begin{lemma} \label{long_lem:rainbow}
  If there exists a $(k,d,q,t)$-edifice $\cT$ on universe~$\cX$ that is
  $(b,t')$-wide, then there exists a $(k, b^k(d+k-1), b^{k-1}q,
  t')$-edifice on the same universe~$\cX$.
\end{lemma}
\longonly{
\squeeze{\squeezelen}
\begin{proof}
  The desired edifice is built by ``merging'' certain carefully chosen
  sets of vertices of~$\cT$.

  Define the following colour-and-trim procedure on a vertex~$u$ of~$\cT$.
  If~$u$ is a leaf, then do nothing. Otherwise, let
  $\cV_1,\ldots,\cV_{t'}$ be subsets of $\cC_\cT(u)$ satisfying
  \ref{long_ax:partition}--\ref{long_ax:parallel}.  For each $i\in[t']$, for each
  $v \in \cV_i$, assign colour~$i$ to the edge from~$u$ to~$v$.  Delete
  all uncoloured edges out of~$u$ as well as the subtrees pointed to by
  these edges.  Then recursively colour-and-trim the remaining vertices
  in~$\cC_\cT(u)$. 

  Let~$\cT'$ be the fully edge-coloured $(bt')$-ary tree obtained by
  applying this colour-and-trim procedure to $r$, the root of $\cT$.
  Reusing the varieties from~$\cT$ makes~$\cT'$ a $(k,d,q,bt')$-edifice.

  For each vertex~$v$ of~$\cT'$, define the \emph{rainbow} at~$v$ to be
  the sequence of colours on the unique path from~$r$ to~$v$. Create a
  new edge-coloured rooted tree~$\cT''$ by merging vertices of~$\cT'$
  that have the same rainbow into ``supervertices'' and defining the
  parent of a supervertex~$v''$ to be the vertex~$u''$ whose rainbow is
  obtained by deleting the last colour in the rainbow at~$v''$; assign
  this deleted colour to the edge from~$u''$ to~$v''$.
  Property~\ref{long_ax:width} implies that each~$\cV_i$ is nonempty;
  property~\ref{long_ax:partition} then implies that~$\cT''$ is a $t'$-ary
  tree with~$k$ levels.

  For each vertex~$u''$ of~$\cT''$, let $\class{u''}$ denote the set of
  vertices of~$\cT'$ that were merged to produce~$u''$.
  Define the variety $X_{u''} \subseteq \cX$ thus:
  \[
    X_{u''} = \biguplus_{u\in\class{u''}} X_u \,.
  \]
  By~\ref{long_ax:parallel}, the above union is indeed a {\em disjoint} union
  (denoted by ``$\uplus$'').
  
  We shall show that~$\cT''$ is the desired edifice. Properties
  \ref{long_ax:arity}, \ref{long_ax:depth}, and~\ref{long_ax:variety} are immediate.
  Property~\ref{long_ax:incidence} follows from the same property of~$\cT'$
  and the observation that whenever two vertices of~$\cT'$ are merged in
  $\cT''$, so are their parents. For property~\ref{long_ax:sink}, first note
  that \ref{long_ax:width} implies that at each level $j\in[k]$ there are
  exactly $b^{k-j}$ vertices of~$\cT'$ that have a particular rainbow.
  Thus, for each leaf~$z''$ of~$\cT''$ we have $|\class{z''}| = b^{k-1}$.
  Using property~\ref{long_ax:sink} of~$\cT'$, we have
  \[
    |X_{z''}| = \bigg| \biguplus_{z\in\class{z''}} X_z \bigg|
    = \sum_{z\in\class{z''}} |X_z| \ge \sum_{z\in[z']} q
    = b^{k-1}q \,.
  \]

  Finally, we address property~\ref{long_ax:intersection}. As in the proof of
  \Cref{long_thm:edifice}, it suffices to upper-bound $|X_{z''} \cap
  X_{u''}|$, where~$z''$ is a leaf of~$\cT''$ and~$u''$ is a vertex of~$\cT''$
that is not an ancestor of~$z''$, whereas the parent~$y''$ of~$u''$ is.
Suppose that~$u''$ is at level $j < k$. Then
  \begin{equation}
    X_{z''} \cap X_{u''}
    = \bigg( \bigcup_{z\in\class{z''}} X_z \bigg) \cap
      \bigg( \bigcup_{u\in\class{u''}} X_u \bigg)
    = \bigcup_{z\in\class{z''},~ u\in\class{u''}} (X_z \cap X_u) \,.
    \label{long_eq:cable-intersection}
  \end{equation}
  Since $|\class{z''}| = b^{k-1}$ and $|\class{u''}| = b^{k-j}$, this latter
  expression immediately leads to $|X_{z''} \cap X_{u''}| \le
  b^{2k-j-1}(d+k-1)$, using property~\ref{long_ax:intersection} of $\cT'$.
  However, this upper bound is too weak; to strengthen
  it, we consider the structure of~$X_{z''}$ and~$X_{u''}$ more
  carefully.

  Consider a generic $z\in\class{z''}$ and a generic $u\in\class{u''}$.
  There must exist $y_1,y_2\in\class{y''}$ such that~$z$ is a descendant
  of~$y_1$ and~$u$ is a descendant of~$y_2$. The crucial observation is
  that if $y_1 \ne y_2$, then by~\ref{long_ax:parallel}, $X_{y_1} \cap
  X_{y_2} = \varnothing$, whence by~\ref{long_ax:incidence},
$X_z \cap X_u = \varnothing$. Therefore
  the pair $(z,u)$ contributes to the latter union in
  \cref{long_eq:cable-intersection} only when $y_1 = y_2$. Therefore,
  \[
    X_{z''} \cap X_{u''}
    = \bigcup_{y\in\class{y''}} \bigcup_{
        \ontopt{z\in\class{z''},~u\in\class{u''}}
        {z,u~\text{descendants of}~y}}
      (X_z \cap X_u) \,.
  \]
  Since $|\class{y''}| = b^{k-j-1}$ and each $y\in\class{y''}$ has $b^j$
  descendants in $\class{z''}$ and $b$ descendants in $\class{u''}$, we
  obtain $|X_{z''} \cap X_{u''}| \le b^{k-j-1} b^j b (d+k-1) =
  b^k(d+k-1) \le b^k(d+k-1)+k-1$, as required.
\end{proof}
}

\begin{lemma} \label{long_lem:is-wide}
  The $(k,d,q,t)$-edifice constructed in \Cref{long_thm:edifice} is
  $(\floor{\delta q}, \floor{1/\delta}t/q)$-wide for all $\delta\in(0,1]$.
\end{lemma}
\longonly{
\squeeze{\squeezelen}
\begin{proof}
  It suffices to prove the lemma in the case $\delta = 1$; a little
  thought shows that the general case then follows as a corollary.

  Let~$\cT$ be the edifice constructed in \Cref{long_thm:edifice}. Let~$u$ be
  a non-leaf vertex of~$\cT$, at level $j+1$, where $j \in [k-1]$. Then
  the edges out of~$u$ are labelled by the~$t$ distinct rank-$j$
  edificial equations. Let us call two such equations
  $\ed{\ell_j}{f_{k-j}}$ and $\ed{\ell^+_j}{f^+_{k-j}}$ {\em similar} if
  $\ell_j = \ell^+_j$ and $f_{k-j} - f^+_{k-j}$ is a constant
  polynomial. This similarity relation then naturally extends to
  $\cC_\cT(u)$. Similarity is easily seen to be an equivalence relation,
  each of whose equivalence classes has size exactly $|\FF_q| = q$.
  Therefore there are exactly~$t/q$ equivalence classes; let
  $\cV_1,\ldots,\cV_{t/q} \subseteq \cC_\cT(u)$ be these classes.

  To show that~$\cT$ is $(q,t/q)$-wide, we shall show that these classes
  $\{\cV_i\}$ satisfy properties \ref{long_ax:partition}, \ref{long_ax:width},
  and~\ref{long_ax:parallel}.  The first two properties are immediate. For
  the third, consider arbitrary $v \ne v' \in \cV_i$, for some~$i$.
  Then~$v$ and~$v'$ are similar, which means that a point $\bx =
  (x,y_1,\ldots,y_{k-1}) \in X_v \cap X_{v'}$ must satisfy a pair of
  similar, but distinct, edificial equations. Let these equations be
  $\ed{\ell_j}{f_{k-j}}$ and $\ed{\ell_j}{f^+_{k-j}}$.
  Consulting \cref{long_eq:edificial}, we find that
  \[
    0 = y_j - y_j
    = \ell_j(y_1,\ldots,y_j,f_{k-j}(x)) - \ell_j(y_1,\ldots,y_j,f^+_{k-j}(x))
    = \ell_j(0, \ldots, 0, f_{k-j}(x) - f^+_{k-j}(x)) \,.
  \]
  By definition, the linear form $\ell_j(z_1,\ldots,z_j)$ has a nonzero
  $z_j$-coefficient, implying that $f_{k-j}(x) - f^+_{k-j}(x) = 0$.
  This is a contradiction, because $f_{k-j} - f^+_{k-j}$ is a {\em
  nonzero} constant polynomial.  Therefore such a point~$\bx$ does not
  exist, i.e., $X_v \cap X_{v'} = \varnothing$.
\end{proof}
}

\shortonly{
\begin{theorem} \label{long_thm:lb-partial}
  Let $c > 1$ be a constant. Let~$\cA$ be a $p$-pass streaming algorithm
  with the following guarantee. For all large enough~$n$ and~$m$ and all
  $\eps \in (0,\frac12]$, for all instances of $\partcover_{n,m,\eps}$, with
  probability at least $2/3$, $\cA$ returns the {\em value} of some
  $\alpha$-approximate solution to the instance, where
  $\alpha < \min\{n^{1/(p+1)}, \eps^{-1/p}\}/(8c(p+1)^2)$.
  %\begin{equation} \label{long_eq:lb-partial}
  %\end{equation}
  Then~$\cA$ cannot be semi-streaming: it needs $\Omega(n^c/p^3)$ bits of space.
\end{theorem}
\squeeze{.7\squeezelen}
\begin{proof}[Proof sketch]
  Combining \Cref{long_thm:edifice} with \Cref{long_lem:rainbow,long_lem:is-wide} and
  working through some algebra, we find that there exists a $(p+1,
  (\delta q)^{p+1}(d+p), (\delta q)^p q, \floor{1/\delta}
  q^{d+p}(1-1/q))$-edifice~$\cT$ over a universe~$\cX$ with $|\cX| = n =
  q^{p+1}$.  Here, $q$ is a large prime power and $\delta \approx
  (2\eps)^{1/p}$.  Using this edifice in a reduction from \mpj
  analogously to \Cref{long_thm:reduction,long_thm:lb} gives us the claimed bound.
\end{proof}
}
\longonly{
\begin{theorem} \label{long_thm:lb-partial}
  Let $c > 1$ be a constant. Let~$\cA$ be a $p$-pass streaming algorithm
  with the following guarantee. For all large enough~$n$ and~$m$ and all
  $\eps \in (0,\frac12]$, for all instances of $\partcover_{n,m,\eps}$, with
  probability at least $2/3$, $\cA$ returns the {\em value} of some
  $\alpha$-approximate solution to the instance, where
  \begin{equation} \label{long_eq:lb-partial}
    \alpha < \frac{\min\{n^{1/(p+1)}, \eps^{-1/p}\}} {8c(p+1)^2} \,.
  \end{equation}
  Then~$\cA$ must use $\Omega(n^c/p^3)$ bits of space. In particular~$\cA$ cannot be semi-streaming.
\end{theorem}
\squeeze{\squeezelen}
\begin{proof}
  This theorem is analogous to a combination of
  \Cref{long_thm:reduction,long_thm:lb}; the proof is along very similar lines.

  We may as well assume that $\eps^{-1/p} \le n^{1/(p+1)}$, because
  if~$\eps$ is too small for this to hold, then we simply consider the
  weaker problem of $(1-\eps')$-partial covering, where $\eps' =
  n^{-p/(p+1)}$.  
 
  Pick a sufficiently large prime power~$q$.  Put $n = q^{p+1}$, $d =
  (c-1)(p+1)+1$, $\tdelta = (2\eps)^{1/p}$, and $\delta = \ceil{\tdelta
  q}/q$.  By our assumption, we have $\tdelta \ge 1/q$ and $\tdelta \le
  \delta \le 2\tdelta$.

  Combining \Cref{long_thm:edifice} with \Cref{long_lem:rainbow,long_lem:is-wide} and
  working through some algebra, we find that there exists a $(p+1,
  (\delta q)^{p+1}(d+p), (\delta q)^p q, \floor{1/\delta}
  q^{d+p}(1-1/q))$-edifice~$\cT$ over a universe~$\cX$ with $|\cX| = n$.
  Using the varieties of~$\cT$, we encode each instance~$\pi$ of
  $\mpj_\cT$ as a collection~$\cI(\pi)$ of subsets of~$\cX$ exactly as
  in \Cref{long_thm:reduction} and treat~$\cI(\pi)$ as an instance of
  $\partcover_{n,m,\eps}$. As before, if $\mpj_\cT(\pi) = 1$, then
  $\cI(\pi)$ admits a total cover using $Q_1 := p+1$ sets.

  For the case $\mpj_\cT(\pi) = 0$, we refine the argument used for
  \Cref{long_thm:reduction} as follows. Let~$X_1$ be the variety of~$\cT$ at
  the unique leaf, $v_1$, in $\cT|_\pi$.  As before, the elements of~$X_1$
cannot be covered by sets corresponding to ancestors of~$v_1$,
  and each of the remaining sets in $\cI(\pi)$ can cover at most
  $(\delta q)^{p+1}(d+p)$ such elements. Every $(1-\eps)$-partial cover
  must, in particular, cover at least $|X_1| - \eps |\cX|$ elements of
  $X_1$.  It follows that the cheapest such partial cover uses at least
  $Q_0 := (|X_1| - \eps |\cX|)/((\delta q)^{p+1}(d+p))$ sets. Now,
  \begin{align}
    \frac{Q_0}{Q_1}
    = \frac{|X_1| - \eps |\cX|}{(\delta q)^{p+1}(p+d)(p+1)}
    &\ge \frac{(\delta q)^p q - \eps q^{p+1}}{(\delta q)^{p+1}(p+d)(p+1)} \label{long_eq:leaf-card} \\
    &= \frac{(\delta q)^p q - \frac12 \tdelta^p q^{p+1}}{(\delta q)^{p+1} c(p+1)^2} \notag \\
    &\ge \frac{1}{2\delta c(p+1)^2} \label{long_eq:round-up} \\
    &\ge \frac{1}{4 (2\eps)^{1/p} \cdot c(p+1)^2} 
    \ge \frac{\eps^{-1/p}}{8c(p+1)^2} \,, \label{long_eq:round-apx}
  \end{align}
  where \eqref{long_eq:leaf-card} uses the parameters of the edifice~$\cT$,
  \eqref{long_eq:round-up} uses $\tdelta \le \delta$, and
  \eqref{long_eq:round-apx} uses $\delta \le 2\tdelta$.

  Therefore, \cref{long_eq:lb-partial} gives $\alpha < Q_0/Q_1$. As in
  \Cref{long_thm:reduction}, with an approximation this good, $\cA$ can be
  used to determine $\mpj_\cT(\pi)$ and must consequently use
  $\Omega(t/p^3)$ bits of space, where $t$ is the arity of $\cT$. Since
  $t = \floor{1/\delta} q^{d+p}(1-1/q) = \Omega(n^c)$, this space lower
  bound is $\Omega(n^c/p^3)$.
\end{proof}
}

\section{Discussion}
\label{long_sec:discussion}

\longonly{
We conclude with a more technically detailed description of selected
results from previous work, with the goal of shedding more light on some
of our own results.

%it's worth recapping. Also, see the table in Demaine.}
%As we have noted, the greedy algorithm is implementable as a
%semi-streaming algorithm, but doing so might require as many as
%$\Omega(n)$ passes.
%

In the external-memory setting, without a streaming
restriction, an \emph{eager} implementation of the greedy algorithm
involves an inverted index and a priority queue of set sizes.
Unfortunately, this involves arbitrary (non-local) memory
accesses, leading to poor performance.

Relaxing the strict greedy requirement,
Cormode, Karloff, and Wirth add a set to the solution if its contribution is at
least~$1/\beta$ times the best~\cite{CormodeKW10}.
So that all disk accesses are sequential, initially they allocate sets
to ``buckets'' (files) according to their size, with a bucket for each range
$[\beta^{j},\beta^{j+1})$, $i=0,\ldots,\kappa$, where $\kappa =
\max\lfloor\log_\beta |S_i|\rfloor$.
Starting from $j=\kappa$ down to~$0$, as each set in bucket~$j$ is examined, sequentially,
set~$S_i$ is added to $\Sol$ only if its contribution is at least $\beta^j$;
otherwise, $(i,S_i\setminus C)$ is appended to the appropriate bucket.
This is essentially the same thresholding as \Cref{long_alg:basic}, with the same
pass/approximation tradeoff, but implemented so
that the total amount of data handled is $O(\beta/(\beta-1))$ times the input
size.

Blelloch, Simhadri, and Tangwongsan solve very large set cover instances on disk and in parallel
in RAM~\cite{blellochst12}.
They consider situations in which there is less than one word of memory per
element.
Their pre-bucketing is much like the geometric ranges of DFG, and their MaNIS
scheme appears to be a randomised, and parallelisable,
version of the pass through the sets in a bucket.
%that are in  covering of
%elements based on set contribution.
%The authors claim considerable speedup over DFG.
%\tony{though I don't believe their numbers!}

The Emek--Ros{\'e}n scheme~\cite{EmekR14}
is in some sense like DFG in its having a hierarchy of thresholds that are
powers of~$2$.
Its purpose however, is to facilitate partial covers with (item and) set costs.
In the unweighted setting,
as each set~$S$ is seen,
it is deemed to cover some subset~$T \subseteq S$,
where $2^i \leq |T| <2^{i+1}$,
if each element in~$T$ was previously covered by some subset
of size $< 2^i$, or was previously uncovered.
This is somewhat like all the runs of DFG with $\beta=2$ being folded into
one.
In parallel, the scheme records the cheapest set that covers each item
(amongst equal-cheapest, choose the first that occurs in the stream).
This step is similar to the folding in \Cref{long_alg:folded}.

We contrast the threshold chosen in our algorithm with that in the
Emek--Ros{\'e}n algorithm.
In our two-pass algorithm (folded into one), $\tau = \sqrt n$, leading to a
$2\sqrt n$ approximation (in fact, $|\Sol| \leq \sqrt n(1+|\Opt|)$).
Once the stream is done, the Emek--Ros{\'e}n algorithm can choose a
threshold~$\tau = 2^i$.
Items that are recorded as covered by some such~$T \subseteq S$,
with $|T|\geq \tau$, are certified to be covered by~$S$; those ``below
the threshold'' are instead covered by their cheapest set.
This way, at most $O(n/\tau)$ $T$-sets are chosen and $O(\tau w(\Opt))$
elements are cheapest-set covered.
Since such a cheapest set has cost at most $w(\Opt)$,
by setting this threshold~$\tau$ to be approximately $\eps n/w(\Opt)$,
the algorithm
returns an $O(w(\Opt)/\eps)$ weight solution.
Of course, we do not know~$w(\Opt)$, but it suffices to choose the largest~$\tau$ so
that at most~$\eps n$ elements are cheapest-set covered.
When $\eps \leq 1/\sqrt n$ however, it is better to choose~$\tau$ to leave at
most~$\sqrt{n}$ cheapest-set covered elements, hence $\tau = \Theta(\sqrt n
/w(\Opt))$.

This tradeoff allows the Emek--Ros{\'e}n algorithm to account for set weights.
In the unweighted case, however, our solution has at most $\sqrt n(1+|\Opt|)$
sets, whereas the Emek--Ros{\'e}n solution has at most $\sqrt n (1+8|\Opt|)$ sets.
As mentioned in \Cref{long_sec:partial-upper-bound}, the latter expression can
become arbitrarily close, i.e., $\sqrt n(1+(1+\delta)^3|\Opt|)$, with space increasing by a
factor of $O(1/\delta)$.

}
 %only in long version
\section*{Acknowledgment}

The second author is grateful to Andrew McGregor for discussions about some recent work.

\bibliographystyle{abbrv}

\end{longversion}

\end{document}